\documentclass[12pt]{article}
\usepackage{amsmath,paralist}
\usepackage[compact]{titlesec}
\usepackage{hyperref}
\usepackage{enumitem}
\usepackage{bbm,color}
\usepackage{amsthm}
\titlespacing{\section}{0pt}{*2}{*1.5}
\titlespacing{\subsection}{0pt}{*2}{*1.5}
\titlespacing{\subsubsection}{0pt}{*2}{*1.5}

\usepackage{amssymb}
\usepackage{graphicx,subcaption}
\usepackage{natbib}
\usepackage[margin = 2cm]{geometry}

\def\EqualDist{\stackrel{\rm d}{=}}

\newcommand{\cov}{\mbox{Cov}}

\newtheorem{lemma}{Lemma}

\newtheorem{definition}{Definition}
\newtheorem{theorem}{Theorem}

\newtheorem{proposition}{Proposition}

\date{}
\author{S. Ward, E. A. K. Cohen\footnote{Correspondence to e.cohen@imperial.ac.uk}, and N. M. Adams \\ Department of Mathematics, Imperial College London, \\ South Kensington Campus, London SW7 2AZ }
\title{Functional summary statistics and testing for independence in marked point processes on the surface of three dimensional convex shapes}
\begin{document}
\maketitle
\begin{abstract}
The fundamental functional summary statistics used for studying spatial point patterns are developed for marked homogeneous and inhomogeneous point processes on the surface of a sphere. These are extended to point processes on the surface of three dimensional convex shapes given the bijective mapping from the shape to the sphere is known. These functional summary statistics are used to test for independence between the marginals of multi-type spatial point processes with methods for sampling the null distribution proposed and discussed. This is illustrated on both simulated data and the RNGC galaxy point pattern, revealing attractive dependencies between different galaxy types.
\end{abstract}

\vspace{0.8cm}

\section{Introduction}\label{Section:1}

The analysis of marked spatial point patterns has focused predominantly on those that arise in a Euclidean space $\mathbb{R}^n$, with a rich literature available for modelling and exploratory data analysis. Central to existing methodologies are the functional summary statistics \citep{Lieshout1999,vanLieshout2006}, a generalisation of their univariate counterparts that analyse the proximity of events of different types to explore if, for example, they are consistent or not with the marginal component processes being independent.

In many applications, marked point patterns (known as multi-type when the mark-space is discrete) are increasingly being captured on surfaces that are not adequately modelled by the geometry of a Euclidean space. For example, terrestrial point patterns occur on a sphere (or more accurately a spheroid), and in microbiology, researchers are interested in the spatial relationship between differently labelled proteins on the bacterial plasma membrane, with the cells shapes more appropriately modelled by ellipsoids or capsules \citep{Kumar2024}. In such applications current methodologies that assume the space is Euclidean would give erroneous results due to their inability to respect and adapt to the underlying geometry of where the objects or events of interest reside.

Advances in the modelling and analysis of point patterns in non-Euclidean spaces are restricted to the univariate (single-type) setting, and have predominantly focused on spheres \citep{Robeson2014,Lawrence2016,Moller2016,Cuevas-Pacheco2018,Jun2019} and linear networks \citep{Ang2012,Rakshit2017,Mcswiggan2017,Moradi2018,Rakshit2019}. The functional summary statistics were extended to general convex bounded shapes in \cite{Ward2020} and non-parametric intensity estimation was developed for Riemannian manifolds in \cite{Ward2023}.

In this paper, we outline how exploratory data analyses can be conducted for marked point processes on the surfaces of convex bounded shapes in $\mathbb{R}^3$. We develop the marked functional summary statistics and demonstrate how they can be used to detect interactions between marginal (component) processes. As detailed in \cite{Ward2020} for the univariate setting of testing for complete spatial randomness, simply substituting the Euclidean metric for a more appropriate one is not a sufficient solution. The surface, in general, does not present the infinite number of isometries needed for the analogous notions of \emph{stationarity} and \emph{isotropy}, which are required for defining the functional summary statistics, e.g. Ripley's $K$-function \citep{Ripley1977}. This problem is retained in the marked setting, and therefore our approach builds on that of \cite{Ward2020}. By first developing the marked functional summary statistics for point patterns on the sphere $\mathbb{S}^2$, we are subsequently able to project the point pattern of interest onto $\mathbb{S}^2$ and conduct our analyses there, exploiting the infinite isometries on the new space. To test for independence between components we consider two methods for sampling from the null distribution of the relevant test statistics. The first is an analogue of the toroidal shift method of \cite{Lotwick1982} that uses random shifts of one of the component patterns to induce independence but preserve the marginal distributions. The second method samples inhomogeneous Poisson patterns from the required intensity function, as mapped onto the sphere.

\section{Preliminaries}\label{Section:2}

\subsection{Geometry}

Our geometric construction follows that of \cite{Ward2020} Let $\mathbf{x}\in\mathbb{R}^3$ such that $\mathbf{x}=(x_1,x_2,x_3)^T$ and define $||\mathbf{x}||=(x_1^2+x_2^2+x_3^2)^{1/2}$ to be the Euclidean norm with the origin of $\mathbb{R}^3$ denoted as $\boldsymbol{0}=(0,0,0)^T$. Denote a subset of $\mathbb{R}^3$ as $\mathbb{D}=\{\mathbf{x}\in\mathbb{R}^3:g(\mathbf{x})=0\}$, where $g:\mathbb{R}^3\mapsto\mathbb{R}$ is the level-set function of $\mathbb{D}$. We require $\mathbb{D}$ be closed and bounded, and define the set $\mathbb{D}_{int}=\{\mathbf{x}\in\mathbb{R}^3:g(\mathbf{x})<0\}$, i.e. the boundary of $\mathbb{D}_{int}$ is $\mathbb{D}$ and we refer to $\mathbb{D}_{int}$ as the interior of $\mathbb{D}$. The set $\mathbb{D}_{int}$ is said to be convex if and only if for all $\mathbf{x},\mathbf{y}\in \mathbb{D}_{int}$ such that $\mathbf{x}\neq\mathbf{y}$ then $\{\mathbf{z}\in\mathbb{R}^3:\mathbf{z}=\mathbf{x}+\gamma(\mathbf{y}-\mathbf{x}),\; \gamma\in (0,1)\}\subset \mathbb{D}_{int}$. We thus define $\mathbb{D}$ to be convex if and only if its interior $\mathbb{D}_{int}$ is convex.

For any closed, bounded and convex $\mathbb{D}$ with level-set function $g$, we define $\tilde{g}$ which rearranges $g(\mathbf{x})=0$, such that $x_3=\tilde{g}(x_1,x_2)$. It may not always be possible to find $\tilde{g}$ explicitly since it may not be a proper function. This issue can be rectified by partitioning $\mathbb{D}$ appropriately. For example take the case of a sphere with radius 1, then $g(\mathbf{x})=x_1^2+x_2^2+x_3^2-1$, hence $\tilde{g}(x_1,x_2)=\pm (1-x_1^2-x_2^2)^{1/2}$, which is not a proper function. In this case we partition the region $\mathbb{D}$ into the regions: $x_3\geq 0$, for which $\tilde{g}(x_1,x_2)=+(1-x_1^2-x_2^2)^{1/2}$, and $x_3<0$, for which $\tilde{g}(x_1,x_2)=-(1-x_1^2-x_2^2)^{1/2}$.

For any close, bounded and convex $\mathbb{D}$, we define its geodesic as the shortest path between two points $\mathbf{x},\mathbf{y}\in\mathbb{D}$ such that every point in the path is also an element of $\mathbb{D}$. The geodesic distance is denoted $d:\mathbb{D}\times\mathbb{D}\mapsto \mathbb{R}_+$, where $\mathbb{R}_+$ is the positive real line including $0$. Thus $\{\mathbb{D},d(\cdot,\cdot)\}$ is a metric space. We will frequently need to evaluate integrals over $\mathbb{D}$, which can be done using its infinitesimal area element defined as,
\begin{equation*}
d\mathbb{D}=\left\{1+\left(\frac{\partial \tilde{g}}{\partial x_1}\right)^2+\left(\frac{\partial \tilde{g}}{\partial x_2}\right)^2\right\}^{1/2}dx_1 dx_2.
\end{equation*}
We assume, without loss of generality, that $\mathbf{0}\in \mathbb{D}_{int}$. In this setting, we then say the space $\mathbb{D}$ is centred. Our methodology can easily be adapted for non-centred spaces by making the appropriate translations to bring the origin inside $\mathbb{D}$.

\subsection{Marked point process on $\mathbb{D}$}

To define a marked spatial point process on a closed, bounded and convex $\mathbb{D}$, we adopt the notation used in the Euclidean construction of \cite{Moller2004}. Let $\lambda_{\mathbb{D}}(\cdot)$ be the Lebesgue measure restricted to the surface of the convex shape $\mathbb{D}$, which is endowed with its geodesic distance and Borel $\sigma$-algebra $\mathcal{B}(\mathbb{D})$. Let $\mathcal{M}$ be a Polish space equipped with its Borel $\sigma$-algebra $\mathcal{B}(\mathcal{M})$. We define the \emph{ground process} $X_g$ to be a univariate locally finite simple (events do not coincide) point process on $\mathbb{D}$. Informally, if for each event $\mathbf{x}_i\in X_g$ there exists a mark $m_i\in\mathcal{M}$, the resulting process $X$ is a marked point process on $\mathbb{D}$ with mark space $\mathcal{M}$. More formally, let $\mathcal{B}(\mathbb{D}\times\mathcal{M})$ be the Borel $\sigma$-algebra on the product space of $\mathbb{D}$ and $\mathcal{M}$ and define $N_{lf}=\{x\subset \mathbb{D}\times\mathcal{M}: |x\cap A|<\infty \text{ for all bounded } A\in\mathcal{B}(\mathbb{D}\times\mathcal{M}) \}$, where $|\cdot|$ denotes cardinality. Define $\mathcal{N}$ to be the $\sigma$-algebra generated by the sets $\{x\in N_{lf}: |x\cap A|=n\}$ for all bounded $A\in \mathcal{B}(\mathbb{D}\times \mathcal{M})$ and $n\in\mathbb{N}$. Then a marked point process on $\mathbb{D}$ with mark space $\mathcal{M}$ is a measurable mapping from some probability space $(\Omega,\mathcal{F},\mathbb{P})$ into the measurable space $(N_{lf},\mathcal{N})$. The special case of multi-type point processes arises when $\mathcal{M}$ is finite, i.e. $\mathcal{M}=\{1,2,\dots,k\}$ for $k$ a positive integer.

In order to integrate over $\mathbb{D}\times\mathcal{M}$ we require a reference measure over $\{\mathbb{D}\times\mathcal{M},\mathcal{B}(\mathbb{D}\times\mathcal{M})\}$. For this purpose we take $\lambda_{\mathbb{D}}\otimes\nu$ to be our reference measure where $\nu$ is some suitably chosen reference measure on the mark space. Examples of $\nu$ can be probability measures over compact subspaces of $\mathbb{R}$ or a counting measure when $\mathcal{M}$ is finite  \citep[see e.g.][]{chiu2013stochastic,Daley2010}.

\subsection{Intensity measure and function}
Let $D \times M \in\mathcal{B}(\mathbb{D}\times\mathcal{M})$ then we define $N_X(D\times M)$ to be the counting measure of $X$. Defined as
\begin{equation*}
N_X(D\times M) = \sum_{(\mathbf{x},m)\in X}\mathbbm{1}\{(\mathbf{x},m)\in D\times M\},
\end{equation*}
it counts the number of points in $X$ with spatial location lying in $D$ and mark in $M$. The expectation of $N_X$ is denoted the intensity measure: $\mu(D\times M)=\mathbb{E}\{N_X(D\times M)\}$. Throughout this paper we shall make the further assumption that $\mu$ is absolutely continuous with respect to $\lambda_{\mathbb{D}}\otimes \nu$ which, by the Radon-Nikodyn theorem, implies there exists a unique $\rho:\mathbb{D}\times\mathcal{M}\mapsto\mathbb{R}^+$, termed the intensity function of $X$, such that
\begin{equation*}
\mu(D\times M) = \int_{D}\int_M \rho(\mathbf{x},m)\lambda_{\mathbb{D}}(d\mathbf{x})\nu(dm).
\end{equation*}
By construction, the ground process $X_g$ has intensity measure $\mu_g(D)=\mu(D\times \mathcal{M})$ for $D\subset\mathbb{D}$ and so for any fixed $M\subset\mathcal{M}$ we have that $\mu(D\times M)\leq \mu(D\times \mathcal{M}) = \mu_g(D)$. Therefore $\mu$ is absolutely continuous with respect to $\mu_g$, and for any fixed $M\subset \mathcal{M}$,
\begin{equation*}
\mu(D\times M) = \int_{D} K^{\mathbf{x}}(M) \mu_g(d\mathbf{x}).
\end{equation*}
Here, $K^{\mathbf{x}}$ is a probability measure, interpreted as the probability of a given event $\mathbf{x}\in X_g$ having a mark in $M$. Further, by assumption that $\mu$ is absolutely continuous with respect to $\lambda_{\mathbb{D}}\otimes\nu$ we have that $\mu_g$ is absolutely continuous with respect to $\lambda_{\mathbb{D}}$ and thus there exists the intensity function $\rho_g:\mathbb{D}\mapsto\mathbb{R}^+$. If we also suppose that $K^{\mathbf{x}}$ is absolutely continuous with respect to $\nu$ then
\begin{equation*}
\mu(D\times M) = \int_{D} \int_M f^{\mathbf{x}}(m)\rho_g(\mathbf{x}) \nu(dm)\lambda_{\mathbb{D}}(d\mathbf{x}),
\end{equation*}
where $f^{\mathbf{x}}$ is the density of $K^{\mathbf{x}}$ with respect to $\nu$, and hence $\rho(\mathbf{x},m)=f^{\mathbf{x}}(m)\rho_g(\mathbf{x})$.

Marked point process $X$ is homogeneous if $\rho_g(\mathbf{x})=\rho\in\mathbb{R}_+$. In this setting, the intensity function of $X$ is given by $\rho(\mathbf{x},m)= f^{\mathbf{x}}(m)\rho$, where $f^{\mathbf{x}}$ is the density of the first order mark distribution with respect to the reference measure $\nu$ over $\mathcal{M}$. We say that $X$ has \emph{common mark distribution} if $K^{\mathbf{x}}\equiv K$ is independent of $\mathbf{x}$, and thus $f^{\mathbf{x}}\equiv f$: the density of $K$ with respect to $\nu$. Further to this, if $K$ coincides with $\nu$ then we have that $f\equiv 1$ and so $\rho(\mathbf{x},m)=\rho_g(\mathbf{x})$. We also define the \emph{unit rate (marked) Poisson process} to be a marked process which has a ground Poisson process with intensity one almost everywhere and marks independent of location and independently and identically distributed with probability measure $P_\mathcal{M}$ over a mark space $\mathcal{M}$. Additionally if $P_\mathcal{M}$ admits a density $p_{\mathcal{M}}$ with respect to the reference measure $\nu$ then the unit rate Poisson process has intensity $\rho(\mathbf{x},m)=p_\mathcal{M}(m)$.

\subsection{Higher order moments}
Higher order intensity functions are further defined as densities of \emph{factorial moment measures} with respect to the $n$-fold reference measure $(\lambda_{\mathbb{D}}\otimes \nu)^n$. The $n^{th}$-order factorial moment measure of $X$ is defined as
\begin{equation*}
\alpha^{(n)}(B_1,\dots,B_n)=\mathbb{E}\left[ \mathop{\sum\nolimits\sp{\ne}}_{(\mathbf{x}_1,m_1),\dots,(\mathbf{x}_n,m_n)\in X}\mathbbm{1}\left\{(\mathbf{x}_1,m_1)\in B_1,\dots,(\mathbf{x}_n,m_n)\in B_n\right\}\right],
\end{equation*}
where $B_i\in \mathcal{B}(\mathbb{D}\times \mathcal{M})$ $(i=1,\dots,n)$ and $\sum^{\neq}$ is the sum over pairwise distinct elements. We note, $\alpha^{(1)}=\alpha=\mu$. If $\alpha^{(n)}$ is absolutely continuous with respect to $(\lambda_{\mathbb{D}}\otimes \nu)^n$, by \emph{Campbell's formula} there exists densities $\rho^{(n)}:(\mathbb{D}\times\nu)^n\mapsto\mathbb{R}^+$ such that for any measurable function $f:(\mathbb{D}\times\nu)^n\mapsto\mathbb{R}^+$,
\begin{multline*}
\mathbb{E}\left[\mathop{\sum\nolimits\sp{\ne}}_{(\mathbf{x}_1,m_1),\dots,(\mathbf{x}_n,m_n)\in X}f\{(\mathbf{x}_1,m_1),\dots,(\mathbf{x}_n,m_n)\}\right] =\\
\int_{\mathbb{D}\times \mathcal{M}}\cdots\int_{\mathbb{D}\times \mathcal{M}}f\{(\mathbf{x}_1,m_1),\dots,(\mathbf{x}_n,m_n)\}\rho^{(n)}\{(\mathbf{x}_1,m_1),\dots,(\mathbf{x}_n,m_n)\}\prod_{i=1}^n\lambda_{\mathbb{D}}(d\mathbf{x}_i)\nu(dm_i),
\end{multline*}
where $\rho^{(n)}\{(\mathbf{x}_1,m_1),\dots,(\mathbf{x}_n,m_n)\} \prod_{i=1}^n\lambda_{\mathbb{D}}(d\mathbf{x}_i)\nu(dm_i)$ can be heuristically considered as the joint probability of finding events in infinitesimal areas $(d\mathbf{x}_i,dm_i)$ $(i=1,\dots,n)$. In the present paper, we assume the existence of $\rho^{(n)}$ for all $n$.
Generalising the $n=1$ case, we define $K^{\mathbf{x}_1,\dots,\mathbf{x}_n}$ to be the $n^{th}$-order mark distributions, and existence of $\rho^{(n)}$ gives $\rho^{(n)}_g$ the $n^{th}$-order intensity functions of $X_g$. Again, assuming $K^{\mathbf{x}_1,\dots,\mathbf{x}_n}$ is absolutely continuous with respect to the $n$-fold measure $\nu^{n}$ with density $f^{\mathbf{x}_1,\dots,\mathbf{x}_n}:\mathcal{M}^n\mapsto\mathbb{R}^+$,
\begin{equation*}
\rho^{(n)}\left\{(\mathbf{x}_1,m_1),\dots,(\mathbf{x}_n,m_n)\right\} = f^{\mathbf{x}_1,\dots,\mathbf{x}_n}(m_1,\dots,m_n)\rho_g^{(n)}\left\{(\mathbf{x}_1,m_1),\dots,(\mathbf{x}_n,m_n)\right\}.
\end{equation*}
Further to this we say that $X$ is independently marked if for all $n\in\mathbb{N}$ the densities $f^{\mathbf{x}_1,\dots,\mathbf{x}_n}$ can be given as the product of their marginals, i.e. $f^{\mathbf{x}_1,\dots,\mathbf{x}_n}(m_1,\dots,m_n)=\prod_{i=1}^n f^{\mathbf{x}_i}(m_i)$.

The pair correlation function (PCF) of $X$ at two points $(\mathbf{x},m_{\mathbf{x}}),(\mathbf{y},m_{\mathbf{y}})\in \mathbb{D}\times \mathcal{M}$ is defined as
\begin{equation*}
g\left\{(\mathbf{x},m_{\mathbf{x}}),(\mathbf{y},m_{\mathbf{y}})\right\}=\frac{\rho^{(2)}\left\{(\mathbf{x},m_{\mathbf{x}}),(\mathbf{y},m_{\mathbf{y}})\right\}}{\rho(\mathbf{x},m_{\mathbf{x}})\rho(\mathbf{y},m_{\mathbf{y}})}.
\end{equation*}
When $X_g$ is a Poisson process and the markings in $X$ are independent, $g\left\{(\mathbf{x},m_{\mathbf{x}}),(\mathbf{y},m_{\mathbf{y}})\right\}=1$ for all $(\mathbf{x},m_{\mathbf{x}})$ and $(\mathbf{y},m_{\mathbf{y}})$. This model typically serves as a benchmark of no interaction between events. When $g\left\{(\mathbf{x},m_{\mathbf{x}}),(\mathbf{y},m_{\mathbf{y}})\right\}$ is greater than one, this indicates clustering between these mark types, and when less than one it indicates inhibition.

The $n^{th}$-order correlation functions $\xi_n$ for $n\in\mathbb{N}$ are defined recursively \citep[e.g.][]{White1979,vanLieshout2006} and based on the $n^{th}$-order intensity functions. Set $\xi_1=1$, then for $n\geq 2$,
\begin{equation}
\label{eq:n:correl:fun}
\sum_{k=1}^l\sum_{E_1,\dots, E_l}\prod_{j = 1}^k \xi_{|E_j|}\left[\{(\mathbf{x}_i,m_i):i\in E_j\}\right]=\frac{\rho^{(n)}\left\{(\mathbf{x}_1,m_1),\dots,(\mathbf{x}_n,m_n)\right\}}{\rho(\mathbf{x}_1,m_1)\cdots\rho(\mathbf{x}_n,m_n)},
\end{equation}
where $\sum_{E_1,\dots, E_l}$ is the sum over all possible $l$-sized partitions of the set $\{1,\dots,n\}$ such that $E_j\neq \emptyset$.

\subsection{Reduced Palm process}
The summary statistics developed in the present paper are based on the \emph{reduced Palm process} of a marked point process $X$, defined via the Radon-Nikodyn derivative of the \emph{reduced Campbell measure},
$$C^!(A \times N)=\mathbb{E}\left(\sum_{(\mathbf{x},m)\in X}\mathbbm{1}\left[\left\{(\mathbf{x},m),X\setminus(\mathbf{x},m)\right\}\in A \times N \right]\right),$$
where $A\in\mathcal{B}(\mathbb{D}\times \mathcal{M}),$ and $N\in\mathcal{N}$, with respect to $\mu$ \citep[e.g. Appendix C.2][]{Moller2004}. In other words, since $C^!$ is absolutely continuous with respect to $\mu$ we have by the Radon-Nikodyn theorem
\begin{equation*}
C^!(A \times N) = \int_{B\times M} P^{!}_{(\mathbf{x},m)}(N)\mu(d\mathbf{x},dm),
\end{equation*}
where $(\mathbf{x},m)\in\mathbb{D}\times\mathcal{M}$, $P^{!}_{(\mathbf{x},m)}$ is the Radon-Nikodyn derivative $dC^!/d\mu$ and defines a probability measure called the reduced Palm measure \citep[e.g. Appendix C.2][]{Moller2004}. We frequently use $P^{!}_{(\mathbf{x},m)}(N)$ as a short hand for $P(X^{!}_{(\mathbf{x},m)}\in N)$ for $N\in\mathcal{N}$. The process $X^{!}_{(\mathbf{x},m)}$ that follows $P^{!}_{(\mathbf{x},m)}$ is referred to as the reduced Palm process of $X$. Heuristically, $P^{!}_{(\mathbf{x},m)}$ can be considered the conditional distribution of $X$ given that $(\mathbf{x},m)\in X$. Based on this we also have the \emph{Campbell-Mecke formula}
\begin{equation*}
\mathbb{E}\left[\mathop{\sum\nolimits\sp{\ne}}_{(\mathbf{x},m)\in X}f\left\{(\mathbf{x},m),X\setminus (\mathbf{x},m)\right\}\right] = \int_{\mathbb{D}\times \mathcal{M}}\mathbb{E}^{!}_{(\mathbf{x},m)}[f\left\{(\mathbf{x},m),X\right\}]\mu(d\mathbf{x},dm),
\end{equation*}
where $\mathbb{E}^{!}_{(\mathbf{x},m)}$ is the expectation under the measure $P^{!}_{(\mathbf{x},m)}$. Following \cite{Cronie2016} we can also define the \emph{$\nu$-averaged Palm measure} with respect to $M\in\mathcal{B}(\mathcal{M})$ where $\nu(M)>0$ as,
\begin{equation}\label{eq:nu:average:palm}
P^{!}_{\mathbf{z},M}(N) = \frac{1}{\nu(M)}\int_M P^{!}_{(\mathbf{x},m)}(N)\nu(dm), \quad N\in\mathcal{N}.
\end{equation}
$P^{!}_{\mathbf{z},M}$ defines a probability measure since $0\leq P^{!}_{(\mathbf{x},m)}(\cdot) \leq 1$ and it may be interpreted as the conditional distribution of $X$ given it has an event at $\mathbf{x}$ with mark in $M$. We frequently use $P^{!}_{\mathbf{x},M}(N)$ as a short hand for $P(X^{!}_{\mathbf{x},M}\in N)$ for $N\in\mathcal{N}$.

In the case of $X$ being a multi-type point process with $\mathcal{M}=\{1,\dots,k\}$ and reference measure $\nu$, then for $M=\{i\}$ for some $i\in\{1,\dots ,k \}$, we have
\begin{equation*}
P^{!}_{\mathbf{x},M}(N) = \frac{1}{\nu(i)}\nu(i)P^{!}_{(\mathbf{x},i)}(N)=P^{!}_{(\mathbf{x},i)}(N).
\end{equation*}
 From (\ref{eq:nu:average:palm}), and identically to \cite[][Appendix A.2]{Iftimi2019}, the expectation with respect to $P^!_{(\mathbf{x},m)}$ is
\begin{equation}\label{eq:expect:mark:set}
\mathbb{E}^{!}_{\mathbf{x},M}\left\{f(X)\right\} = \frac{1}{\nu(M)}\int_M \mathbb{E}^{!}_{(\mathbf{x},m)}\left\{f(X)\right\}\nu(dm),
\end{equation}
for measurable, non-negative functions $f$.

\subsection{Generating functional}
In order to define upcoming summary functional statistics we make use of the probability generating functional (pgf) which uniquely characterises a point process $X$. Let $u:\mathbb{D}\times\mathcal{M}\mapsto[0,1]$ be a measurable function with bounded support then the pgf $G(\cdot)$ is defined as \citep{Moller2004}
\begin{equation*}
G(u) = \mathbb{E}\left[\prod_{(\mathbf{x},m)\in X}\left\{1-u(\mathbf{x},m)\right\}\right].
\end{equation*}
Further to this, given that the $n^{th}$-order intensity functions exist for all $n$ and that the following series is convergent, the generating functional can be represented as an infinite series \citep{Cronie2016},
\begin{equation*}
G(u) = 1+\sum_{n=1}^{\infty}\int_{\mathbb{D}\times\mathcal{M}}\cdots\int_{\mathbb{D}\times\mathcal{M}}\rho^{(n)}\left\{(\mathbf{x}_1,m_1),\dots,(\mathbf{x}_n,m_n)\right\}\prod_{i=1}^n u(\mathbf{x}_i,m_i)\lambda_{\mathbb{D}}(d\mathbf{x}_i)\nu(m_i),
\end{equation*}
where we take the convention that an empty product is 1. Further to this, using (\ref{eq:expect:mark:set}) we can define the generating functional $G^!_{\mathbf{x},M}$ with respect to $P^!_{\mathbf{x},M}$ as,
\begin{equation*}
G^!_{\mathbf{x},M}(u) =\frac{1}{\nu(M)}\int_M \mathbb{E}^{!}_{(\mathbf{x},m)}\left[\prod_{(\mathbf{y},n)\in X}\left\{1-u(\mathbf{y},n)\right\}\right]\nu(dm).
\end{equation*}

\section{Summary statistics for isotropic processes on the sphere}\label{Section:3}

\subsection{Isotropic processes}

In keeping with the formulation of \cite{Moller2016}, let $\mathcal{O}(3)$ denote the set of $3\times 3$ rotation matrices and for $O\in\mathcal{O}(3)$, let $OX=\{(O\mathbf{x},m):(\mathbf{x},m)\in X\}$ be the rotation of point process $X$ by $O$.

\begin{definition}
	\label{def:iso}
Point process $X$ on $\mathbb{S}^2$ is isotropic if $OX$ and $X$ are identically distributed for any $O\in\mathcal{O}(3)$.
\end{definition}

Isotropy of $X$ implies $\mu(D\times M)=\mu(OD\times M)$ for any $M\subset\mathcal{M}$, $D\subseteq\mathbb{S}^2$ and $O\in\mathcal{O}(3)$. Rotational invariance further implies $\mu(D\times M)=\eta_M\lambda_{\mathbb{S}^2}(D)$ where $\eta_M$ is a positive constant potentially depending on $M$. If $X$ is isotropic then so is the ground process $X_g$ which consequently has constant intensity $\rho_g$ and $\mu(D\times M)=\int_D K^{\mathbf{x}}(M) \rho_g \lambda_{\mathbb{S}^2}(d\mathbf{x})$. Therefore,
$$
0 = \int_{D}\eta_M - K^{\mathbf{x}}(M) \rho_g \lambda_{\mathbb{S}^2}(d\mathbf{x}),
$$
and since this must hold for any $D\subseteq\mathbb{S}^2$, it follows that $\eta_M = K^{\mathbf{x}}(M) \rho_g$ and the mark distribution is independent of $\mathbf{x}$. This allows the relaxation of notation to write $\mu(D\times M) = \rho_g K(M)\lambda_{\mathbb{S}^2}(D)$. In the special case of isotropic processes on a sphere, we can take the reference measure $\nu$ over $\mathcal{M}$ to be $K$, which we refer to as the \emph{canonical mark measure} on the sphere.  Applying this in (\ref{eq:nu:average:palm}) gives the isotropic spheroidal counterpart to $\nu$-average Palm measures, as defined for stationary marked processes in $\mathbb{R}^n$ by \cite{Cronie2016}.

If $X$ is isotropic, then it can be shown that $OX^!_{(\mathbf{x},m)}$ and $X^!_{(O\mathbf{x},m)}$ are identically distributed for any $O\in\mathcal{O}(3)$. This is analogous to the equivalent translational result in \cite{Moller2004}. Moreover defining the origin on $\mathbb{S}^2$ to be the North pole, i.e. $\mathbf{o}=(0,0,1)^T$ and $O_\mathbf{x}$ to be the unique rotation orthogonal to the geodesic between $\mathbf{x}$ and $\mathbf{o}$ such that $O_\mathbf{x}\mathbf{o}=\mathbf{x}$, then $O_{\mathbf{x}}X^!_{(\mathbf{o},m)}$ and $X^!_{(\mathbf{x},m)}$ are identically distributed. This is the marked version of \cite[Proposition 1][]{Moller2016} and can easily be proven under the mild condition of $X$ being absolutely continuous to the unit rate Poisson over mark space $\mathcal{M}$. Therefore, from (\ref{eq:nu:average:palm}),
\begin{equation*}
P(X^{!}_{\mathbf{o},M}\in \cdot) = P\left\{O_{\mathbf{x}}^{T}(X^{!}_{\mathbf{x},M})\in \cdot\right\},
\end{equation*}
for almost all $\mathbf{x}\in\mathbb{S}^2$.

\subsection{Functional summary statistics}
\label{subsec:func}
The functional summary statistics for isotropic marked point processes are an extension of their Euclidean counterparts as defined in \cite{vanLieshout2006}.
\begin{definition}
	\label{def:fss_iso}
Let $C,E\subseteq \mathcal{M}$ then
\begin{align}
F^E(r) &= P\left[X\cap \left\{B_{\mathbb{S}^2}(\mathbf{o},r)\times E\right\}\neq \emptyset\right]\label{eq:iso:sphere:F}\\
D^{CE}(r) &= P[X^!_{\mathbf{o},C}\cap\left\{B_{\mathbb{S}^2}(\mathbf{o},r)\times E\right\}\neq \emptyset]\label{eq:iso:sphere:D}\\
J^{CE}(r) &= \frac{1-D^{CE}(r)}{1-F^E(r)} \quad \text{for } F^E(r)<1 \label{eq:iso:sphere:J}\\
K^{CE}(r) &= \frac{1}{\lambda_{\mathbb{S}^2}(A)\nu(C)\nu(E)}\mathbb{E}\left[\mathop{\sum\nolimits\sp{\ne}}_{(\mathbf{x},m_{\mathbf{x}}),(\mathbf{y},m_{\mathbf{y}})\in X}\frac{\mathbbm{1}\left\{(\mathbf{x},m_\mathbf{x})\in A\times C, (O_{\mathbf{x}}^T\mathbf{y},m_\mathbf{y})\in B_{\mathbb{S}^2}(\mathbf{o},r)\times E\right\}}{\rho(\mathbf{x},m_{\mathbf{x}}) \rho(\mathbf{y},m_{\mathbf{y}})}\right],\label{eq:iso:sphere:K}
\end{align}
for all $r\in[0,\pi]$, and where $K^{CE}(\cdot)$ does not depend on $A\subseteq\mathbb{S}^2$.
\end{definition}
Commonly, $F^E(\cdot)$ is referred to as the empty space function, and $D^{CE}(\cdot)$ is the cross nearest neighbour distribution function. The cross $K$ function $K^{CE}(\cdot)$ is a spheroidal and marked version of the \emph{multi-type cross $K$ function} \citep[e.g.][]{Moller2004}, a multivariate extension of the classical Ripley's $K$-function.

Under the isotropic assumption and taking the canonical mark measure $\nu=K$, we obtain $\rho(\mathbf{x},m)= f^{\mathbf{x}}(m)\rho_g$, where $\rho_g$ is the constant intensity of the ground process. Thus, this can replace $\rho(\mathbf{x},m)$ in (\ref{eq:iso:sphere:K}). Furthermore, isotropy implies all four functional summary statistics are independent of the typical point $\mathbf{o}$. For example,
\begin{align*}
F^{E}(r;\mathbf{x}) &= P[X\cap \left\{B_{\mathbb{S}^2}(\mathbf{x},r)\times E\right\}\neq \emptyset]\\
&=P[O_{\mathbf{x}}^TX\cap \{B_{\mathbb{S}^2}(O_{\mathbf{x}}^T\mathbf{x},r)\times E\}\neq \emptyset]\\
&=P[X\cap \left\{B_{\mathbb{S}^2}(\mathbf{o},r)\times E\right\}\neq \emptyset]\\
&= F^{E}(r),
\end{align*}
with a similar argument holding for $D^{CE}(\cdot)$ and hence $J^{CE}(\cdot)$. The invariance of $K^{CE}(\cdot)$ to the typical point is shown in Appendix \ref{Kinvariance} by representing it as an expectation with respect to the reduced Palm distribution. It is important to note that in general $J^{CE}(\cdot)\neq J^{EC}(\cdot)$ whilst $K^{CE}(\cdot)=K^{EC}(\cdot)$.

Let $C, E\subset \mathcal{M}$ with $C\cap E =\emptyset$, and define $X_C= X\cap (\mathbb{S}^2\times C)$ and $X_E= X\cap (\mathbb{S}^2\times E)$. The following proposition describes how the functional summary statistics behave under independence of $X_C$ and $X_E$, and can be considered marked spheroidal analogues to the multi-type Euclidean results given by \cite{Lieshout1999}.

\begin{proposition}\label{prop:iso:function}
	Let $X$ be a marked isotropic spheroidal point process. Let $C,E\subset \mathcal{M}$ such that $C\cap E =\emptyset$ and suppose that $X_C$ and $X_E$ are independent. For $r\in[0,\pi]$,
	\begin{align*}
	D^{CE}(r) &= F^{E}(r),\\
	J^{CE}(r) &= 1,\\
	K^{CE}(r) &= 2\pi\{1-\cos(r)\}.
	\end{align*}
\end{proposition}

Proof: See Appendix \ref{proof:iso:function}.

These functional summary statistics can be interpreted in an analagous manner to the summary statistics given in \cite{Moller2004} for stationary Euclidean multi-type processes. If $K^{CE}(r)>2\pi\{1-\cos(r)\}$ and $J^{CE}(r)<1$ then this indicates that points with marks in $C$ aggregate around points with marks in $E$, whilst if $K^{CE}(r)<2\pi\{1-\cos(r)\}$ and $J(r)>1$ then they repel.

\subsection{Estimating functional summary statistics}
Following \cite{vanLieshout2006}, supposing that the process is observed on some window $W\subset \mathbb{S}^2$, such that $\lambda_{\mathbb{S}^2}(W)>0$, and that the reference measure and probability measure over $\mathcal{M}$ coincide (thus $\rho(\mathbf{x},m)=\rho_g$ where $\rho_g$ is the constant intensity of the ground process), then estimators are defined as
\begin{align*}
1-\hat{D}^{CE}(r) &= \frac{1}{\rho_g \lambda_{\mathbb{S}^2}(W_{\ominus r})\nu(C)}\\ & \qquad \times\sum_{(\mathbf{x},m_{\mathbf{x}})\in X} \mathbbm{1}\{(\mathbf{x},m_{\mathbf{x}})\in W_{\ominus r}\times C\}\prod_{(\mathbf{y},m_{\mathbf{y}})\in X}\left[1- \mathbbm{1}\{d_{\mathbb{S}^2}(\mathbf{x},\mathbf{y}) < r,m_{\mathbf{y}}\in E\}\right],\\
1-\hat{F}^{E}(r) &=\frac{1}{|I_{W_{\ominus r}}|}\sum_{\mathbf{p}\in I_{W_{\ominus r}}} \prod_{(\mathbf{x},m)\in X} \left[1-\mathbbm{1}\left\{d_{\mathbb{S}^2}(\mathbf{p},\mathbf{x})\leq r, m\in E\right\}\right],\\
\hat{J}^{CE}(r) &= \frac{1-\hat{D}^{CE}(r)}{1-\hat{F}^{E}(r)}\quad \text{for } \hat{F}^{E}(r)< 1,\\
\hat{K}^{CE}(r) &= \frac{1}{\rho_g^2\lambda_{\mathbb{S}^2}(W_{\ominus r})\nu(C)\nu(E)} \\ & \qquad\times\sum_{(\mathbf{x},m_{\mathbf{x}})\in X}\sum_{(\mathbf{y},m_{\mathbf{y}})\in X\setminus (\mathbf{x},m_{\mathbf{x}})}\mathbbm{1}\{(\mathbf{x},m_\mathbf{x})\in W_{\ominus r}\times C, d_{\mathbb{S}^2}(\mathbf{x},\mathbf{y})\leq r, m_{\mathbf{y}}\in E\},
\end{align*}
where $W_{\ominus r}$ is the erosion of $W$ by distance $r\in\mathbb{R}_+$ and $I_{W_{\ominus r}} = I\cap W_{\ominus r}$, where $I$ is a finite grid of points on $\mathbb{S}^2$. 

\begin{proposition}\label{prop:iso:bias}
	Let $X$ be a marked isotropic point process on $\mathbb{S}^2$. Let $C,E\subset \mathcal{M}$ such that $C\cap E =\emptyset$ and suppose that $X_C$ and $X_E$ are independent. The estimators $\hat{F}^{E}(\cdot)$, $\hat{D}^{CE}(\cdot)$ and $\hat{K}^{CE}(\cdot)$ are unbiased, whilst $\hat{J}^{CE}(\cdot)$ is ratio-unbiased when $\rho_g$ is known.
\end{proposition}

Proof: See Appendix \ref{proof:iso:bias}.

Note that when the process is completely observed over $\mathbb{S}^2$, erosion of the space is not needed and $W_{\ominus r}$ can be replaced by $\mathbb{S}^2$. Furthermore, it may be that $\nu$ is unknown and must be estimated. One approach is to use $\widehat{\nu(C)} = N_X(W\times C)/N_X(W\times \mathcal{M})$ as a plug in estimator. Using the canonical measure $\nu=K$ results in $\nu$ being a probability measure with $\nu(\mathcal{M})=1$. In this case $\mathbb{E}[N_X(W\times C)]=\rho_g\lambda_{\mathbb{S}^2}(W)\nu(C)$ and $\mathbb{E}[N_X(W\times \mathcal{M})]=\rho_g\lambda_{\mathbb{S}^2}(W)\nu(\mathcal{M})=\rho_g\lambda_{\mathbb{S}^2}(W)$ and hence $\widehat{\nu(C)}$ is ratio unbiased.

 For the discrete mark case of $\mathcal{M}=\{1,2,\dots,k\}$, let $\nu$ be the counting measure. Then $X$ has first order mark distribution $K(i)$ with respect to $\nu$. Let us define $ \rho_i\equiv\rho(\mathbf{x},i)=\rho_g K(i)$. Taking $C=\{i\}$ and $E=\{j\}$ ($i\neq j$), our estimators then simplify to
\begin{align*}
1-\hat{D}^{ij}(r) &= \frac{1}{\rho_i\lambda_{\mathbb{S}^2}(W_{\ominus r})}\sum_{\mathbf{x}\in X_{i}} \mathbbm{1}(\mathbf{x}\in W_{\ominus r})\prod_{\mathbf{y}\in X_{j}}\left[1- \mathbbm{1}\{d_{\mathbb{S}^2}(\mathbf{x},\mathbf{y}) < r\}\right],\\
1-\hat{F}^{j}(r) &=\frac{1}{|I_{W_{\ominus r}}|}\sum_{\mathbf{p}\in I_{W_{\ominus r}}} \prod_{\mathbf{x}\in X_j} \left[1-\mathbbm{1}\{d_{\mathbb{S}^2}(\mathbf{p},\mathbf{x})\leq r\}\right],\\
\hat{J}^{ij}(r) &= \frac{1-\hat{D}^{ij}(r)}{1-\hat{F}^{j}(r)},\quad \text{for } \hat{F}^{j}(r)< 1,\\
\hat{K}^{ij}(r) &= \frac{1}{\rho_i \rho_j\lambda_{\mathbb{S}^2}(W_{\ominus r})}\sum_{\mathbf{x}\in X_{i}}\sum_{\mathbf{y}\in X_j}\mathbbm{1}\{\mathbf{x}\in W_{\ominus r}, d_{\mathbb{S}^2}(\mathbf{x},\mathbf{y})\leq r\}.
\end{align*}
In this setting, when the discrete mark measure $K$ is unknown, then using the same reasoning as the general setting,  we can use $\widehat{K(i)}=N_{X_i}(W)/N_{X_g}(W)$ for $i \in \{1,2,\dots,k\}$.
\subsection{Testing for independence}
\label{test_isotropic}
Consider the null hypothesis $H$ that states $X_C$ and $X_E$ are independent. To test this using estimates of the functional summary statistics, we need to determine whether deviations from their theoretical value under $H$ are significant. The typical approach in the Euclidean setting is to compare these estimates against a sample constructed from simulates under the null hypothesis of independence \citep{Moller2004,Myllymaki2013}. Simulating from the null is complicated by the fact that its distribution is dependent on the marginal distributions of the $X_C$ and $X_E$, which are often unknown.

There are two common approaches used for stationary multi-type processes in $\mathbb{R}^d$, each of which can be adapted to the spheroidal setting. The toroidal shift approach developed by \cite{Lotwick1982} bootstraps the observed patterns, allowing for the unknown marginal distributions to be approximately maintained while breaking dependence though random shifts of the patterns. This can be readily adapted to multi-type spheroidal patterns where we instead randomly rotate on the sphere. A restriction to this rotation method is it requires a point pattern to have been observed over all of the sphere.

The alternative approach simulates from a specified null model, e.g. Poisson process \citep[e.g.][]{Rajala2018}. While this imposes additional distribution assumptions on the null hypothesis, which may be misspecified, previous works \citep{Wiegand2012,Moller2004,Rajala2018} have used this approach to provide evidence of dependency between components in a multi-type Euclidean process. Such an approach can be used here, even in the event of a partially observed spheroidal process. We explore both approaches through simulations and a real data study.

\section{Summary statistics for inhomogeneous processes on the sphere}
\label{inhomo_sphere}

\subsection{Intensity reweighted isotropic processes}

To proceed with developing functional summary statistics for inhomogeneous processes on the sphere, we must assume the process admits isotropic structure once reweighted with respect the intensity function.
\begin{definition}
	\label{def:IRWMI}
	Let $X$ be a simple marked processes on $\mathbb{S}^2$ whose intensity functions of all orders exist, and define $\bar{\rho}_E=\inf_{\mathbf{x}\in\mathbb{S}^2,m\in E}\rho(\mathbf{x},m)$. If $\bar{\rho}\equiv\bar{\rho}_{\mathcal{M}}>0$ and, for all $n\geq 1$, $(\mathbf{x}_i,m_i)\in\mathbb{S}^2\times\mathcal{M}$ and $O\in\mathcal{O}(3)$, the $n^{th}$-order correlation functions are isotropic in the sense
\begin{equation*}
\xi_{n}\{(\mathbf{x}_1,m_1),\dots,(\mathbf{x}_n,m_n)\}=\xi_{n}\{(O\mathbf{x}_1,m_1),\dots,(O\mathbf{x}_n,m_n)\},
\end{equation*}
 then $X$ is intensity-reweighted momented isotropic (IRWMI).
 \end{definition}
 This is the isotopy analogue of the intensity-reweighted momented stationary definition used by \cite{VanLieshout2011,Cronie2016} for inhomogeneity in the Euclidean setting. It is also the marked extension of the IRWMI definition in \cite{Ward2020} for unmarked spheroidal processes. It should be noted that any isotropic process is immediately IRWMI. Other IRWMI processes include (inhomogeneous) Poisson processes, multi-type processes with independent IRWMI components, or a ground process that is unmarked IRWMI \citep[in the sense of][]{Ward2020} with events having independent marks.

A weaker condition is an analogue of second order intensity reweighted stationary processes, defined in \cite{Cronie2016} for marked processes in Euclidean spaces.
\begin{definition}
	\label{def:SOIRWI}
	Let $X$ be a simple marked processes on $\mathbb{S}^2$. Let $\mathcal{K}^{CE}(B)$ for $A,B\subset \mathbb{S}^2$ and $C,E\subseteq \mathcal{M}$ be defined as
	\begin{equation}\label{eq:second:order:measure}
	\lambda_{\mathbb{S}^2}(A)\nu(C)\nu(E)\mathcal{K}^{CE}(B) =
	\mathbb{E}\left[\mathop{\sum\nolimits\sp{\ne}}_{(\mathbf{x},m_{\mathbf{x}}),(\mathbf{y},m_{\mathbf{y}})\in X}\frac{\mathbbm{1}\left\{(\mathbf{x},m_\mathbf{x})\in A\times C, (O_{\mathbf{x}}^T\mathbf{y},m_\mathbf{y})\in B\times E\right\}}{\rho(\mathbf{x},m_{\mathbf{x}}) \rho(\mathbf{y},m_{\mathbf{y}})}\right].
	\end{equation}
	If $\mathcal{K}^{CE}$ does not depend on $A$ then $X$ is second-order intensity reweighted isotropic (SOIRWI).
\end{definition}
It can be shown using the Campbell formula that an isotropic PCF is sufficient for Definition \ref{def:SOIRWI} to hold. By this, we mean $g(\mathbf{x},m_{\mathbf{x}},\mathbf{y},m_{\mathbf{y}})=g(O\mathbf{x},m_{\mathbf{x}},O\mathbf{y},m_{\mathbf{y}})$ for any $O\in\mathcal{O}(3)$. Considering $n=2$ in (\ref{eq:n:correl:fun}), it follows that $\xi_2\{(\mathbf{x},m_{\mathbf{x}}),(\mathbf{y},m_{\mathbf{y}})\} + 1 = g\{(\mathbf{x},m_{\mathbf{x}}),(\mathbf{y},m_{\mathbf{y}})\}$, and therefore SOIRWI is a weaker condition than IRWMI.

\subsection{Inhomogeneous cross $K$ function}
Similar to \citep[Definiton 4.8]{Moller2004}, we define an analogous inhomogeneous cross $K$ function for marked point patterns on a sphere.
\begin{definition}
	\label{def:K:inhom}
		Let $X$ be a marked SOIRWI processes on $\mathbb{S}^2$ with intensity function $\rho(\mathbf{x},m)$, where $\bar\rho>0$. Let $C,E\subset\mathcal{M}$, then for $r\in[0,\pi]$,
	\begin{equation*}
		K^{CE}_{\rm inhom}(r) = \mathcal{K}(B_{\mathbb{S}^2}(\mathbf{o},r)),
	\end{equation*}
	where $\mathcal{K}$ is defined in Definition \ref{def:SOIRWI}.
\end{definition}

The following propositions show that the inhomogeneous cross $K$ function has the same properties as their isotropic counterparts under the assumption of independence.

\begin{proposition}\label{prop:K:inhom:equal:2pi}
	Let $X$ be a marked SOIRWI point process on $\mathbb{S}^2$ and consider two disjoint sets $C,E\subset\mathcal{M}$ such that they have positive measure $\nu$. Further assume that $X_C$ and $X_E$ are independent. Then,
	\begin{equation*}
	K^{CE}_{\rm inhom}(r) = 2\pi\{1-\cos(r)\}.
	\end{equation*}
\end{proposition}
Proof: Under isotropy of the PCF and using the Campbell formula it can be shown that
\begin{equation*}
\lambda_{\mathbb{S}^2}(A)\nu(C)\nu(E)\mathcal{K}^{CE}(B) = \int_{C}\int_{B\times E} g\{(\mathbf{o},m_\mathbf{x}),(\mathbf{y},m_{\mathbf{y}})\} \nu(dm_{\mathbf{x}})d\lambda_{\mathbb{S}^2}(\mathbf{y})\nu(dm_{\mathbf{y}}),
\end{equation*}
where $\mathbf{o}$ is some arbitrary point in $\mathbb{S}^2$. Setting $B=B_{\mathbb{S}^2}(\mathbf{o},r)$ and noting that by the independence assumption $g\{(\mathbf{x},m_{\mathbf{x}}),(\mathbf{y},m_{\mathbf{y}})\}\equiv 1$ the proposition holds.

We again maintain the heuristic understanding that if $K^{CE}_{\rm inhom}(r)>2\pi\{1-\cos(r)\}$ then this indicates that points with marks in $C$ aggregate around points with marks in $E$, whilst they repel if $K^{CE}_{\rm inhom}(r)<2\pi\{1-\cos(r)\}$.

\subsection{Inhomogeneous cross nearest neighbour distribution and empty space functions}
The following definition extends the empty space function $F^E(r)$ and cross nearest neighbour distribution $D^{CE}(r)$ presented in Section \ref{subsec:func} to inhomogeneous processes. These definitions provide the spheroidal counterpart to those introduced by \cite{Cronie2016} for Euclidean point processes.

\begin{definition}
	\label{def:fss}
	Let $X$ be a marked IRWMI processes on $\mathbb{S}^2$ with intensity function $\rho(\mathbf{x},m)$, where $\bar\rho>0$. Let $C,E\subset\mathcal{M}$, and
	\begin{equation}\label{eq:u}
	u^r_{\mathbf{y},E}(\mathbf{x},m) = \frac{\bar{\rho}_E\mathbbm{1}\{(\mathbf{x},m)\in B_{\mathbb{S}^2}(\mathbf{y},r)\times E\}}{\rho(\mathbf{x},m)}, \quad \mathbf{y}\in\mathbb{S}^2,
	\end{equation}
	then for $r\in[0,\pi]$,
	\begin{align*}
	F^{E}_{\rm inhom}(r) & = 1 - G(1-u^r_{\mathbf{y},E}),	\\
	D^{CE}_{\rm inhom}(r) & = 1 - G^!_{\mathbf{y},C}(1-u^r_{\mathbf{y},E}).\\
	\end{align*}
\end{definition}

Definition \ref{def:fss} implicitly depends on the point $\mathbf{y}$, but Theorem \ref{thm:J:inhom:generating} will show that for IRWMI processes these functions are independent of the choice of $\mathbf{y}$.

The following proposition shows that when the process is isotropic $F^{E}_{\rm inhom}$ and $D^{CE}_{\rm inhom}$ collapse to the isotropic setting.

\begin{proposition}\label{prop:equiv:iso:inhom}
Let $X$ be an isotropic spheroidal point process, then $F^{E}_{\rm inhom}(r)$ and $D^{CE}_{\rm inhom}(r)$ as defined in Definition \ref{def:fss} are equal to those in Definition \ref{def:fss_iso}.
\end{proposition}

Proof: See Appendix \ref{proof:equiv:iso:inhom}.

\subsection{Inhomogeneous cross $J$ function}
To define the inhomogeneous $J$ function, and following the formulation of \cite{Cronie2016}, we first note that in the isotropic setting, the $J$ function admits a power series representation
\begin{equation}
\label{eq:JPS}
J^{CE}(r) = \frac{1}{\nu(C)}\left\{\nu(C)+\sum_{n=1}^{\infty}\frac{(-\rho_E)^n}{n!}J^{CE}_n(r)\right\},
\end{equation}
for $r\in[0,\pi]$, where
\begin{multline}
J^{CE}_n(r) = \int\displaylimits_C\phantom{A}\idotsint\displaylimits_{\{B_{\mathbb{S}^2}(\mathbf{o},r)\times E\}^n}  \xi_{n+1}\{(\mathbf{y},m),(O_{\mathbf{y}}\mathbf{x}_1,m_1),\dots,(O_{\mathbf{y}}\mathbf{x}_n,m_n)\} \\ \times\nu(m)\lambda_{\mathbb{S}^2}(d\mathbf{x}_1)\nu(m_1)\cdots\lambda_{\mathbb{S}^2}(d\mathbf{x}_n)\nu(m_n),
\label{eq:Jn}
\end{multline}
for $\mathbf{y}\in\mathbb{S}^2$. This is the spheroidal analogue of the infinite series expansion for the $J^{CE}$ function given by \cite{vanLieshout2006} for Euclidean processes. For isotropic process $X$, and for all $n\geq 1$, $J^{CE}_n(r)$ will be independent of the choice of $\mathbf{y}$, and by (\ref{eq:JPS}), so will $J^{CE}(r)$. In fact, $J^{CE}_n(r)$ will also be independent of the choice of $\mathbf{y}$ for marked IRWMI processes, and therefore we can define the inhomogeneous $J$ function in the same way.
\begin{definition}
		Let $X$ be a marked IRWMI processes on $\mathbb{S}^2$ with intensity function $\rho(\mathbf{x},m)$, where $\bar\rho>0$. Let $C,E\subset\mathcal{M}$ The inhomogeneous cross $J$ function between $C$ and $E$ is defined as
	\begin{equation*}
	J^{CE}_{\rm{inhom}}(r) = \frac{1}{\nu(C)}\left\{\nu(C)+\sum_{n=1}^{\infty}\frac{(-\bar{\rho}_E)^n}{n!}J^{CE}_n(r)\right\},
	\end{equation*}
	for $r\in[0,\pi]$ where $J^{CE}_n(r)$ is as in (\ref{eq:Jn}).
\end{definition}

It is clear from this definition that when $X$ is isotropic $J^{CE}_{\rm inhom}\equiv J^{CE}$ by the series representation of the $J^{CE}$ function \cite{vanLieshout2006}.

The following theorem shows the relationship between $F_{\rm inhom}^{E}$, $D_{\rm inhom}^{CE}$ and $J^{CE}_{\rm{inhom}}$ whilst highlighting that the choice of $\mathbf{y}$ in Definition \ref{def:fss} is merely a matter of convenience.

\begin{theorem} \label{thm:J:inhom:generating}
	Let $X$ be a marked IRWMI processes on $\mathbb{S}^2$. Under the further assumption that all the intensity functions of all orders of $X$ exist and that
	\begin{equation*}
	\limsup_{n\rightarrow\infty}\left(\frac{\bar{\rho}_E}{n!}\idotsint\displaylimits_{(B_{\mathbb{S}^2}(\mathbf{o},r)\times E)^n} \frac{\rho^{(n)}(\mathbf{x}_1,m_1),\dots,(\mathbf{x}_n,m_n)}{\rho(\mathbf{x}_1,m_1)\cdots\rho(\mathbf{x}_n,m_n)}\prod_{i=1}^n \lambda_{\mathbb{S}^2}(d\mathbf{x}_i)\nu(m_i)\right)^{1/n}<1,
	\end{equation*}
	then $D_{\rm inhom}^{CE}(\cdot)$, $F_{\rm inhom}^{E}(\cdot)$ and $J_{\rm inhom}^{CE}(\cdot)$ are independent of $\mathbf{y}$ and
	$$
	J^{CE}_{\rm inhom}(r) = \frac{1-D^{CE}_{\rm inhom}(r)}{1-F^E_{\rm inhom}(r)},
	$$
  for all $r\in [0,\pi]$ for which $F^{E}_{\rm inhom}(r)\neq 1$.
\end{theorem}
Proof: See Appendix \ref{proof:J:inhom:generating} which adapts the proof of Theorem 1 in \cite{Cronie2016}.
\begin{proposition}\label{prop:J:inhom:equal:1}
	Let $X$ be a marked IRWMI point process on $\mathbb{S}^2$ and consider two disjoint sets $C,E\subset\mathcal{M}$ such that they have positive measure $\nu$. Further assume that $X_C$ and $X_E$ are independent and that the assumptions of Theorem \ref{thm:J:inhom:generating} hold. Then,
	\begin{equation*}
	D^{CE}_{\rm inhom}(r) = F^{E}_{\rm inhom}(r)
	\end{equation*}
	and so $J^{CE}_{\rm inhom}(r)\equiv 1$.
\end{proposition}
Proof: see Appendix \ref{proof:J:inhom:equal:1}, which adapts the proof of Proposition 2 in \cite{Cronie2016}.

\subsection{Estimating functional summary statistics}
\label{subsec:est:inhomo}

In order to estimate the inhomogeneous functional summary statistics we follow \cite{Cronie2016} to obtain estimates for $J^{CE}_{\rm inhom}$ whilst estimates for $K^{CE}_{\rm inhom}$ follow similarly to the isotropic setting. Let us suppose that we have observed our spheroidal point process $X$ through some window $W\subseteq\mathbb{S}^2$. Then we can estimate our functional summary statistics as,

\begin{align}
1-\hat{D}^{CE}_{\rm inhom}(r) &= \frac{1}{\lambda_{\mathbb{S}^2}(W_{\ominus r})\nu(C)}\sum_{(\mathbf{x},m_{\mathbf{x}})\in X} \frac{\mathbbm{1}\{(\mathbf{x},m_{\mathbf{x}})\in W_{\ominus r}\times C\}}{\rho(\mathbf{x},m_\mathbf{x})} \nonumber \\ &  \qquad\qquad\qquad\times\prod_{(\mathbf{y},m_{\mathbf{y}})\in X}\left[1- \frac{\bar{\rho}_E\mathbbm{1}\{d_{\mathbb{S}^2}(\mathbf{x},\mathbf{y}) < r,m_{\mathbf{y}}\in E\}}{\rho(\mathbf{y},m_{\mathbf{y}})}\right]\label{eq:inhom:estimator:D}\\
1-\hat{F}^{E}_{\rm inhom}(r) &=\frac{1}{|I_{W_{\ominus r}}|}\sum_{\mathbf{p}\in I_{W_{\ominus r}}} \prod_{(\mathbf{x},m)\in X} \left[1-\frac{\bar{\rho}_E \mathbbm{1}\{d_{\mathbb{S}^2}(\mathbf{p},\mathbf{x})\leq r, m\in E\}}{\rho(\mathbf{x},m_\mathbf{x})}\right]\label{eq:inhom:estimator:F}\\
\hat{J}^{CE}_{\rm inhom}(r) &= \frac{1-\hat{D}^{CE}(r)}{1-\hat{F}^{E}(r)}\quad \text{for } \hat{F}^{E}(r)< 1\label{eq:inhom:estimator:J}\\
\hat{K}^{CE}_{\rm inhom}(r) &= \frac{1}{\lambda_{\mathbb{S}^2}(W_{\ominus r})\nu(C)\nu(E)} \times \nonumber\\ & \qquad\qquad \sum_{(\mathbf{x},m_{\mathbf{x}})\in X}\sum_{(\mathbf{y},m_{\mathbf{y}})\in X\setminus (\mathbf{x},m_{\mathbf{x}})}\frac{\mathbbm{1}\{(\mathbf{x},m_\mathbf{x})\in W_{\ominus r}\times C, d_{\mathbb{S}^2}(\mathbf{x},\mathbf{y})\leq r, m_{\mathbf{y}}\in E\}}{\rho(\mathbf{x},m_\mathbf{x})\rho(\mathbf{y},m_\mathbf{y})}\label{eq:inhom:estimator:K}
\end{align}

\begin{proposition}\label{prop:unbiased:inhom}
	Let $X$ be a marked process on $\mathbb{S}^2$. When the assumptions of Theorem \ref{thm:J:inhom:generating} hold, $\hat{D}^{CE}_{\rm inhom}$ and $\hat{F}^{E}_{\rm inhom}$ are unbiased whilst $\hat{J}^{CE}_{\rm inhom}$ is ratio unbiased for known $\rho$. Furthermore, when $X$ is SOIRWI, $\hat{K}^{CE}_{\rm inhom}$ is unbiased for known $\rho$.
\end{proposition}
Proof: see Appendix \ref{proof:unbiased:inhom} that adapts the proof of Lemma 1 \cite{Cronie2016}.

In most circumstances the intensity function $\rho$ is unknown, requiring a plug-in estimator to be used. When needed, we deploy the kernel intensity estimator of \cite{Ward2023} developed for point processes on Riemannian manifolds.

It may also be the case that $\nu$ is unknown.  \cite{Stoyan2000} advocate replacing $\lambda_{\mathbb{S}^2}(B)\nu(C)$ with
	\begin{equation}\label{eq:stoyan:stoyan:estimator}
	\sum_{(\mathbf{x},m)\in X\cap(B\times C)}\frac{1}{\rho(\mathbf{x},m)}.
	\end{equation}
	The Campbell formula shows this is an unbiased estimator for $\lambda_{\mathbb{S}^2}(B)\nu(C)$. Furthermore, for $\hat{K}^{CE}_{\rm inhom}$,  \cite[][Appendix B]{Iftimi2019} discuss a number of approaches to estimating $\lambda_{\mathbb{S}^2}(B)\nu(C)\nu(E)$. One such approach recognises
	\begin{equation}\label{eq:stoyan:stoyan:nu:estimator}
	\widehat{\nu(C)}=\frac{1}{\lambda_{\mathbb{S}^2}(B)}\sum_{(\mathbf{x},m)\in X\cap(B\times C)}\frac{1}{\rho(\mathbf{x},m)}
	\end{equation}
	is an unbiased estimator for $\nu(C)$.

In practice, our analysis can be re-cast in the terms of multi-type processes, for which marks in $C$ are labelled with $i$ and marks in $E$ are labelled $j$. For this discrete mark case of $\mathcal{M}=\{1,\dots,k\}$,  $\rho(\mathbf{x},i)=\nu(i)\rho_i(\mathbf{x})=\rho_i(\mathbf{x})$ for $i\in\mathcal{M}$. For $i,j\in\mathcal{M}$ $(i\neq j)$, the estimators are
\begin{align*}
1-\hat{D}^{ij}_{\rm inhom}(r) &= \frac{1}{\lambda_{\mathbb{S}^2}(W_{\ominus r})}\sum_{\mathbf{x}\in X_i} \frac{\mathbbm{1}(\mathbf{x}\in W_{\ominus r})}{\rho_i(\mathbf{x})}\prod_{\mathbf{y}\in X_j}\left[1- \frac{\bar{\rho}_j\mathbbm{1}\{d_{\mathbb{S}^2}(\mathbf{x},\mathbf{y}) < r\}}{\rho_j(\mathbf{y})}\right]\\
1-\hat{F}^{j}_{\rm inhom}(r) &=\frac{1}{|I_{W_{\ominus r}}|}\sum_{\mathbf{p}\in I_{W_{\ominus r}}} \prod_{\mathbf{x}\in X_j} \left[1-\frac{\bar{\rho}_j \mathbbm{1}\{d_{\mathbb{S}^2}(\mathbf{p},\mathbf{x})\leq r\}}{\rho_j(\mathbf{x})}\right]\\
\hat{J}^{ij}_{\rm inhom}(r) &= \frac{1-\hat{D}^{ij}(r)}{1-\hat{F}^{j}(r)}\quad \text{for } \hat{F}^{j}(r)< 1\\
\hat{K}^{ij}_{\rm inhom}(r) &= \frac{1}{\lambda_{\mathbb{S}^2}(W_{\ominus r})}\sum_{\mathbf{x}\in X_i}\sum_{\mathbf{y}\in X_j}\frac{\mathbbm{1}\{\mathbf{x}\in W_{\ominus r}, d_{\mathbb{S}^2}(\mathbf{x},\mathbf{y})\leq r\}}{\rho_i(\mathbf{x})\rho_j(\mathbf{y})}.
\end{align*}
Again, plug-in estimators for $\rho_i$ and $\rho_j$ can be computed using the approach of \cite{Ward2023}.

\subsection{Testing for independence}
\label{test_inhom}
Consider again the null hypothesis $H$ that states $X_C$ and $X_E$ are independent. In Section \ref{test_isotropic} we considered random rotations of isotropic patterns to create samples of the functional summary statistics from the null. To handle inhomogeneous patterns on the sphere, we will show we can rotate both the component patterns and the intensity function to sample the functional summary statistics from the null. This is similar to inhomogeneous processes in the Euclidean setting, where \cite{Cronie2016} translate the pattern and the intensity function. 

First consider the random measure $\Xi^X$ which places mass $1/\rho(\mathbf{x},m)$ at points $(\mathbf{x},m)\in X$,
\begin{equation*}
\Xi^X = \sum_{(\mathbf{x},m)\in X} \frac{\delta_{(\mathbf{x},m)}}{\rho(\mathbf{x},m)}.
\end{equation*}
Then for $B\times C\in \mathcal{B}(\mathbb{S}^2\times \mathcal{M})$,
\begin{equation*}
\Xi^X(B\times C) = \sum_{(\mathbf{x},m)\in X} \frac{\mathbbm{1}\{(\mathbf{x},m)\in B\times C\}}{\rho(\mathbf{x},m)}.
\end{equation*}
We define $\Xi_{O}^X(B\times C)$, the rotation of $\Xi^X$ by $O\in\mathcal{O}(3)$ as
$$
\Xi_{O}^X(B\times C) = \Xi^X(OB\times C)= \sum_{(\mathbf{x},m)\in X} \frac{\mathbbm{1}\{(O^{T}\mathbf{x},m)\in B\times C\}}{\rho(\mathbf{x},m)}= \sum_{(\mathbf{x},m)\in O^TX} \frac{\mathbbm{1}\{(\mathbf{x},m)\in B\times C\}}{\rho(O\mathbf{x},m)}.
$$
We say that $\Xi^X$ is isotropic if $\Xi^X\EqualDist \Xi_O^X$ for any $O\in\mathcal{O}(3)$, and if $X$ is IRWMI then the factorial moment measures of $\Xi^X$ and $\Xi_O^X$ are identical for any $O\in\mathcal{O}(3)$.

\begin{proposition}\label{prop:rotation:inhom}
	Let $X$ be a IRWMI marked point process on $\mathbb{S}^2$ for which the assumptions of Theorem \ref{thm:J:inhom:generating} hold and let $C,E$ be disjoint subsets of $\mathcal{M}$ with strictly positive $\nu$ measure. If $X_C$ and $X_E$ are independent and the measure $\Xi^X$ is isotropic then the estimators $\hat D^{CE}_{\rm{inhom}}, \hat F^{E}_{\rm{inhom}}, \hat J^{CE}_{\rm{inhom}},$ and $\hat K^{CE}_{\rm{inhom}}$ can be written as a function of $(\Xi^{X_C},\Xi^{X_E})$ and $(\Xi^{X_C},\Xi_O^{X_E})\EqualDist (\Xi^{X_C},\Xi^{X_E})$.
\end{proposition}
Proof: see Appendix \ref{proof:rotation:inhom} that adapts the proof to \cite[Proposition 4][]{Cronie2016}.

Proposition \ref{prop:rotation:inhom} implies samples can be generated from the null distribution by randomly rotating one or both of $X_C$ and $X_E$, whilst also rotating the intensity function. Although this approach does not require any assumptions on the marginals, it does, as with its isotropic counterpart outlined in Section \ref{test_isotropic}, require the pattern to be completely observed on the sphere. A curiosity when analysing inhomogeneous processes with this random rotation method occurs at $r=\pi$. While, theoretically, $K_{\rm inhom}^{12}(\pi) = 4\pi$ always, in the inhomogeneous case $\hat{K}_{\rm inhom}^{12}(\pi)\neq 4\pi$ and in fact exhibits sampling variance due to its dependence on the pattern. For an observed bivariate point pattern $X=(X_1,X_2)$, suppose that we rotate $X_1$ with a random $O\in\mathcal{O}(3)$ to give a null simulate $\tilde{X} = (OX_1,X_2)$, then by construction $\hat{K}_{\tilde{X}}^{12}(\pi)=\hat{K}_{X}^{12}(\pi)$. Hence all null simulates will converge to the same value at $r=\pi$, which does not reflect the true variance of the estimator. Envelope plots should therefore be treated with caution in this region.

In scenarios where the point process is not completely observed or distant interactions are important we may consider imposing further assumptions on the null (e.g. Poisson) and sampling independent component patterns directly.

\section{Summary statistics for point processes on convex shapes}

\subsection{Mapping spatial point processes from $\mathbb{D}$ to $\mathbb{S}^2$}
Following the discussion given by \cite{Ward2020}, there are a number of impracticalities to defining functional summary statistics directly on a convex shape due to its potential asymmetric nature. This asymmetry means analogous notions of translation and rotational invariance cannot be well-defined, which in turn means functional summary statistics cannot be constructed directly. Instead we follow the procedure of \cite{Ward2020} that maps the point process from its original space onto the sphere where rotational symmetries can now be utilised to construct well-defined functional summary statistics. The following theorem shows that marked point processes on convex shapes can be mapped to a point process on the sphere under an appropriately transformed intensity function.

\begin{theorem}\label{thm:mapping}
	Let $X$ be a marked point process on $\mathbb{D}$ with intensity function $\rho$ and $f:\mathbb{D}\mapsto\mathbb{S}^2$ be a bijective measurable mapping. Then the point process $Y=f(X)\equiv\left[\{f(\mathbf{x}),m\}:(\mathbf{x},m)\in X\right]$, is a marked point process on $\mathbb{S}^2$ with intensity function,
	\begin{equation*}
	\rho^*(\mathbf{x},m) = \rho\{f^{-1}(\mathbf{x}),m\}|J_{f}(\mathbf{x})| \quad \mathbf{x}\in\mathbb{S}^2,
	\end{equation*}
	where $|J_{f}(\mathbf{x})|$ is the determinant of the Jacobian of the transformation $f$.
\end{theorem}

Proof: See for example, \cite[Theorem 5.1]{Last2017}, which can be adapted to the marked point process setting.

\cite{Ward2020} suggest using the pragmatic choice of $f(\mathbf{x}) = \mathbf{x}/\|\mathbf{x}\|$ for the bijective mapping to take the point pattern from $\mathbb{D}$ onto $\mathbb{S}^2$ as this applicable to all convex shapes. It is possible to use simpler alternatives when the geometry of the shape permits, for example when $\mathbb{D}$ is an ellipsoid we may take $f(x_1,x_2,x_3)=(x_1/a,x_1/b,x_1/c)$ as is done in the numerical examples of Section \ref{section:num:examples}. Estimators of the functional summary statistics as given in (\ref{eq:inhom:estimator:D}-\ref{eq:inhom:estimator:K}) are then computed, where instead of $\rho$ we use $\rho^*$ as given by Theorem \ref{thm:mapping}.

When mapping a point process from our convex shape to the sphere, the assumptions of IRWMI and/or SOIRWI of the resulting process may not hold. However, we will show through simulations that the functional summary statistics are still able to capture attractive and repulsive behaviour of a process. We note that there are some special cases in which the process, when mapped from $\mathbb{D}$ to $\mathbb{S}^2$, is known to be IRWMI and/or SOIRWI. These are:
\begin{enumerate}
	\item a Poisson process with an independent marking scheme is IRWMI;
	\item a multi-type Poisson process with independent components is IRWMI;
	\item a multi-type process with independent components is SOIRWI, see e.g \cite[Proposition 4.4][]{Moller2004}.
\end{enumerate}
Special attention should be taken for the Poisson case, since this assumption implies that the mapped spheroidal process is IRWMI. This in turn means we can test the hypothesis of independence using random rotations assuming the process is complete observed, i.e. the Poisson hypothesis is a subset of the IRWMI hypothesis when using random rotations to simulate from the null hypothesis.

\subsection{Computing the functional summary statistics}

To implement the spherical summary functional statistics, we need to construct an estimator of the unknown intensity function once the process has been mapped onto $\mathbb{S}^2$. Let $X$ be a marked point process on $\mathbb{D}$ with intensity function $\rho(\mathbf{x},m)>0$ for $(\mathbf{x},m)\in\mathbb{D}\times\mathcal{M}$ with respect to $\lambda_{\mathbb{D}}$ and some mark reference measure $\nu$. Suppose we have some estimator $\hat{\rho}$ of $\rho$ on the original space $\mathbb{D}$, for example we could use the kernel estimator of \cite{Ward2023}. Then by Theorem \ref{thm:mapping}, the intensity function when mapped to $\mathbb{S}^2$ is estimated as
\begin{equation}
\label{eq:rhoestinhom}
\hat{\rho}^*(\mathbf{x},m) = \hat{\rho}(\mathbf{x},m)|J_{f}(\mathbf{x})|.
\end{equation}

For the discrete mark case $\mathcal{M} = \{1,2,...,k\}$, we can take our reference measure to be the counting measure, giving $\rho(\mathbf{x},i)=\rho_i(\mathbf{x})$ for $i\in\mathcal{M}$. Estimates of each $\rho_i^\ast$ can then be considered individually. When the process is assumed homogenous on $\mathbb{D}$, this is $\hat\rho^\ast_i(\mathbf{x}) = N_{X_i}(\mathbb{D})|J_f(\mathbf{x})|/\lambda_{\mathbb{D}}(\mathbb{D})$. When using an estimator $\hat{\rho}_i(\mathbf{x})$ on $\mathbb{D}$ for the inhomogeneous case it becomes $\hat\rho^\ast_i(\mathbf{x}) = \hat{\rho}_i(\mathbf{x})|J_f(\mathbf{x})|$.

Testing for independence is conducted using the mapped patterns on the sphere, and hence proceeds as described in Section \ref{test_inhom}, with again two possible approaches for simulating from the null.

\section{Numerical examples}\label{section:num:examples}
In the following studies, we consider point patterns sampled from a bivariate point process $X=(X_1,X_2)$ on the surface of an ellipsoid $(x_1,x_2,x_3) = \{a\sin(\theta)\cos(\phi),b\sin(\theta)\cos(\phi),c\cos(\theta)\},$ where $\theta = [0,\pi)$ and $\phi\in [0,2\pi)$. We set $a=b=0.8$ and $c$ to be such that the total surface area equals $4\pi$ ($c=1.44$ to 3.s.f.). We compute functional summary statistics and use these to provide evidence of independence between $X_1$ and $X_2$. In each example, the process being analysed is inhomogeneous with the intensity estimated using the kernel intensity estimator of \cite{Ward2023}. The estimated intensity on the sphere is therefore computed as in (\ref{eq:rhoestinhom}), with $f(x_1,x_2,x_3) = (x_1/a,x_2/b,x_3/c)$ mapping $\mathbb{D}$ to $\mathbb{S}^2$.

In each case, $\hat K^{12}(r)$ is computed, from which we plot $\hat P^{12}(r) = \{\hat K^{12}(r)\}^{1/2} - \left\{2\pi (1-\cos(r))\right\}^{1/2}$. This is the natural estimator for $ P^{12}(r)=\{K^{12}(r)\}^{1/2} - \left\{2\pi (1-\cos(r))\right\}^{1/2}$, the spheroidal analogue to the Euclidean $L$-function \citep{Besag1977} that transforms the $K$-function to give a value of zero for independent marginal processes. We note that, by construction, $K^{12}(r) = K^{21}(r)$, and in the event of a completely observed process over $\mathbb{D}$ so are their estimators, as such we only consider one of these in the analysis. In the event that the process is only partially observed we find that $\hat{K}^{12}(r) = \hat{K}^{21}(r)$ are equivalent only for small $r$, whilst for moderate to large $r$ the window causes minor differences between the estimators. We also plot estimators for $J^{12}(r)$ and $J^{21}(r)$, which in general are not equal to one another.

Given that point patterns are simulated and observed on the whole surface, we elect to implement the random rotation method for sampling the null distribution. The alternative method of sampling from the null, simulating independent Poisson processes, is demonstrated in the real data example of Section \ref{sec:rd}. Simulation envelopes are computed using 199 simulations from the null.
\begin{figure}
	\centering
	\includegraphics[width=0.98\linewidth]{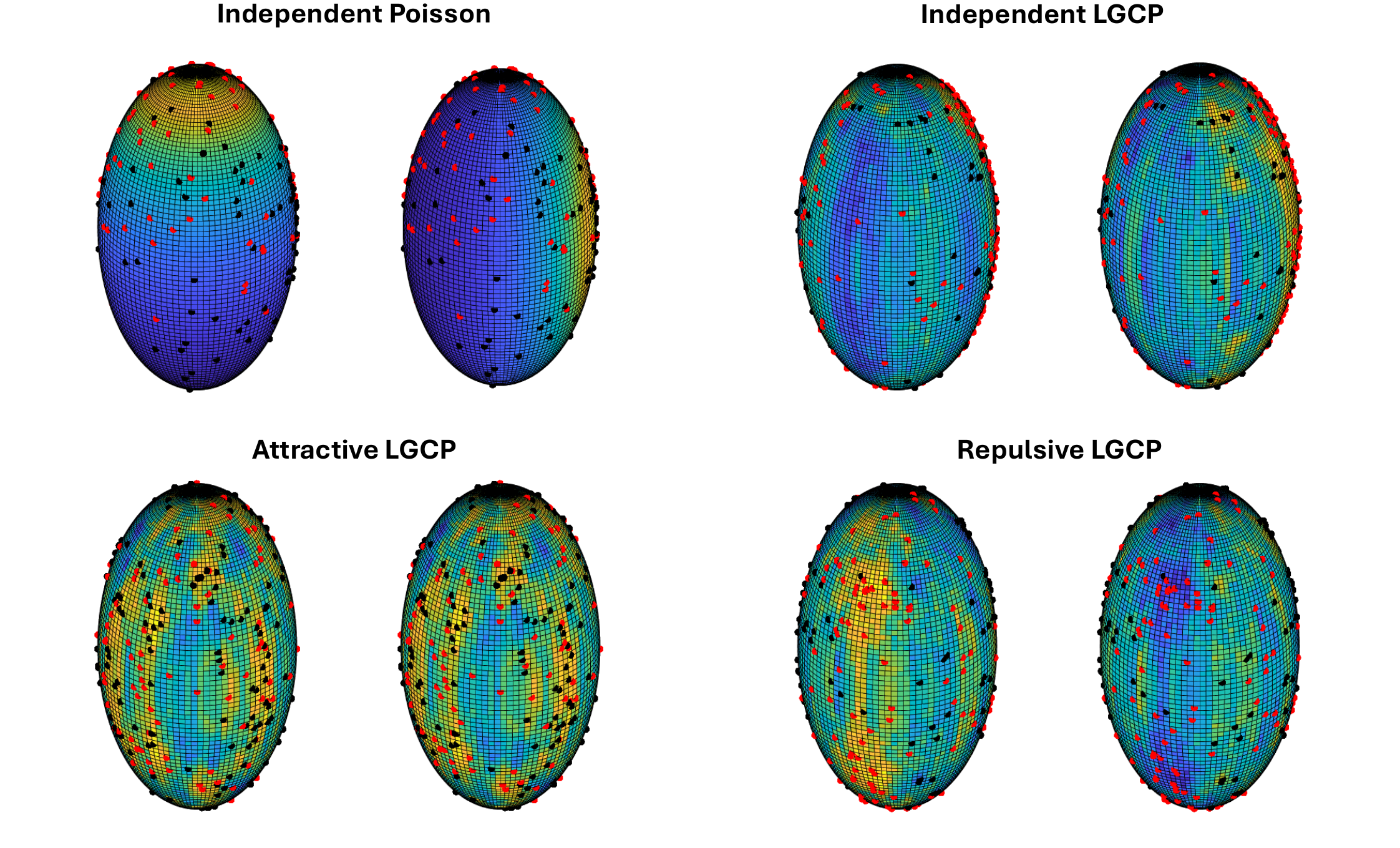}
	\caption{Example realisations of point patterns generated in the numerical examples. Red points are events from $X_1$ and black points are events from $X_2$. For each setting, the point patterns are the same and the ellipsoid is viewed from the same perspective. The heatmap shows the intensity of $X_1$ and $X_2$ in the left and right, respectively for the independent Poisson process example and the conditional GRF in the case of the LGCP examples. \label{fig:ellip}}
\end{figure}

\subsection{Independent Poisson processes on an ellipsoid}
A bivariate Poisson process with independent marginals is simulated on the surface of the ellipsoid using the method of \cite{Ward2023}. The intensities of the marginal processes are $\rho_1(x_1,x_2,x_3) = \exp\{\log(6) + x_3\}$ and  $\rho_2(x_1,x_2,x_3) = \exp\{\log(6)+2x_1\}$, for $\mathbf{x} = (x_1,x_2,x_3)\in\mathbb{D}$. An example realisation of the pattern, along with heat-maps for the intensity functions, is shown in Figure \ref{fig:ellip}. Figure \ref{fig:Pois} demonstrates the functional summary statistics fall within the 95\% simulation envelopes, as expected.
\begin{figure}
	\centering
	\includegraphics[width=0.98\linewidth]{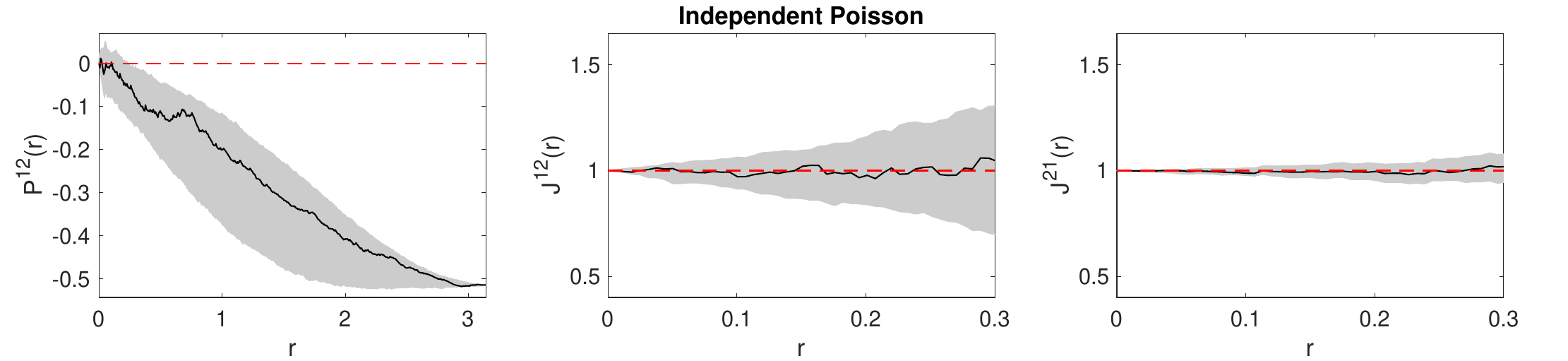}
	\caption{Estimated functional summary statistics $\hat P^{12}(r)$, $\hat J^{12}(r)$ and $\hat J^{21}(r)$ for independent inhomogenous Poisson processes. Grey bands depict the 95\% simulation envelopes for the null distribution using the rotation method. The theoretical value under the null is shown with the red dotted line. \label{fig:Pois}}
\end{figure}

\subsection{Log-Gaussian Cox processes on an ellipsoid}
The bivariate log-Gaussian Cox process (LGCP) is used to model attractive and repulsive interactions between two point processes \citep{Brix2001,Moller2004}. On spheres, univariate LGCP are discussed in depth by \cite{Cuevas-Pacheco2018} whilst multivariate LGCP are discussed by \cite{Jun2019}. Here, we extend the formulation to the ellipsoid $\mathbb{D}$.

To define a bivariate LGCP on $\mathbb{D}$, first define a bivariate Gaussian random field (GRF) on $\mathbb{D}$. The random function $U\equiv(U_1,U_2):\mathbb{D}\times\mathbb{D}\rightarrow \mathbb{R}^2$ is a GRF if for any $n\in\mathbb{N}$, $\mathbf{x}_{1},...,\mathbf{x}_{n}\in\mathbb{D}$ and $\mathbf{a}_{1},...,\mathbf{a}_{n}\in\mathbb{R}^2$, $\sum_{i=1}^n \mathbf{a}_{i}^T U(\mathbf{x}_{i})$ is normally distributed. For bivariate GRF $U$, define the bivariate random field $Z:\mathbb{D}\times\mathbb{D}\rightarrow\mathbb{R}_+^2$ as $Z = \left[\exp\{U_1\},\exp\{U_2\}\right]$. Point process $X=(X_1,X_2)$ is then said to be bivariate LGCP on $\mathbb{D}$ if $X$ given $Z=(z_1,z_2)$ is a bivariate Poisson process with intensity $(z_1,z_2)$.

The GRF is defined via its mean $\mu_i(\mathbf{x}) \equiv E\{U_i(\mathbf{x})\}$ and auto- and cross-covariance functions $c_{ij}(\mathbf{x},\mathbf{y})\equiv\cov\left\{U_i(\mathbf{x}),U_j(\mathbf{y})\right\}$ ($i,j=1,2$). These in-turn completely define the distribution of $X$, whose intensity functions and pair-correlation functions are given as
$$
\rho_i(\mathbf{x}) = \exp\left\{\mu_i(\mathbf{x}) + \frac{c_{ii}(\mathbf{x},\mathbf{x})}{2}\right\}, \qquad g_{ij}(\mathbf{x},\mathbf{y}) = \exp\left\{c_{ij}(\mathbf{x},\mathbf{y})\right\},
$$
respectively, for $i,j = 1,2$.

To simulate a LGCP on the ellipsoid $\mathbb{D}$, we trivially extend to a bivariate setting the procedure provided in the Supplementary Material of \cite{Ward2023} and is also akin to the approach of \cite{Cronie2020} who simulate LGCP on $\mathbb{R}^2$ linear networks by first simulating over $\mathbb{R}^2$ and then restricting the GRF to only those points that coincide with the linear network. Let $\tilde U$ be a bivariate GRF on $\mathbb{R}^3$, then defining $U(\mathbf{x}) = \tilde U(\mathbf{x})$ for $\mathbf{x}\in\mathbb{D}$, $U$ is a GRF on $\mathbb{D}$. Therefore, to simulate a LGCP on $\mathbb{D}$, we simulate $\tilde u$ from $\tilde U$, and take realized GRF $u=(u_1,u_2)$ to be $\tilde u$ restricted to $\mathbb{D}$. A bivariate Poisson process is then simulated on $\mathbb{D}$ with intensity $z=\{\exp(u_1),\exp(u_2)\}$.

A LGCP simulated in this manner will adopt from the original GRF $\tilde U$ the required independence, attractive or repulsive properties. We consider bivariate GRFs $\tilde U$ with covariance functions of the form
\begin{align}\label{eq:cvf1}
		c_{11}(r) = c_{22}(r) & = \sigma^2 \exp(-r/\gamma^2),  \\  \label{eq:cvf2}
		c_{12}(r) = c_{21}(r) & = a_{12}\sigma^2 \exp(-r/\gamma^2),
	\end{align}
noting that since $\tilde U$ is a GRF over $\mathbb{R}^3$ then $r$ in Equations \ref{eq:cvf1} and \ref{eq:cvf2} refer to the Euclidean distance between points. The correlation parameter $a_{12}$ controls the nature of the dependencies between $X_1$ and $X_2$; positive $a_{12}$ gives attractive interactions, negative $a_{12}$ gives repulsive interactions, and $a_{12}=0$ results in $X_1$ and $X_2$ being independent. Parameters $\sigma^2$ and $\gamma^2$ control the overall strength and range of interactions, respectively.

We consider three cases, all of which have covariance functions of the form (\ref{eq:cvf1}) and (\ref{eq:cvf2}), with $\sigma^2=1$ and $\gamma^2 = 0.2$. In the first example, the mean function for $\tilde{U}$ is $\mu_1(x_1,x_2,x_3) = \mu_2(x_1,x_2,x_3) = \log(6)+x_1$, and $a_{12}=0$ to give independence between the two marginal point processes. In the second example, $\mu_1(x_1,x_2,x_3) = \mu_2(x_1,x_2,x_3) = \log(6)+x_1^2$, and $a_{12} = 1$ to give attraction between the marginal point processes. In the third example, we set $\mu_1(x_1,x_2,x_3) = \log(6) + x_2^2$, $\mu_2(x_1,x_2,x_3) = \log(6) + x_1^2$, and $a_{12} = -1$, to give repulsion between the marginal point processes. An example realisation of the pattern, along with heat-maps for the realized intensity functions, is shown in Figure \ref{fig:ellip} The computed functional summary statistics, together with the 95\% simulation envelopes for the null are shown in Figure \ref{fig:LGCP}, displaying the expected behaviour in each case.

\begin{figure}
	\centering
	\includegraphics[width=0.98\linewidth]{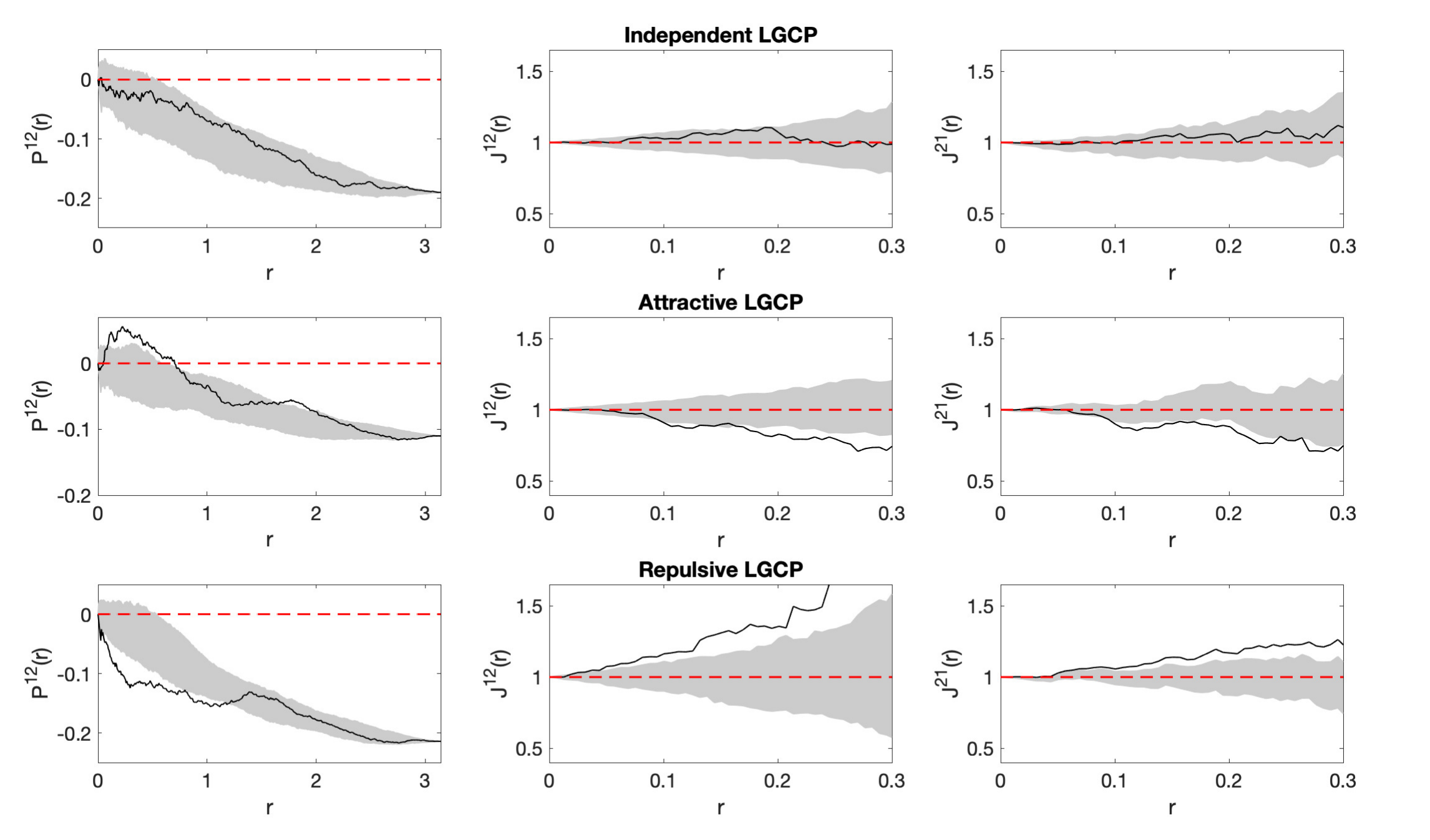}
	\caption{Estimated functional summary statistics $\hat P^{12}(r)$, $\hat J^{12}(r)$ and $\hat J^{21}(r)$ for inhomogeneous LGCPs. Grey bands depict the 95\% simulation envelopes for the null distribution using the rotation method. The theoretical value under the null is shown with the red dotted line.\label{fig:LGCP}}
\end{figure}

\section{Real data example}
\label{sec:rd}
We consider the Revised New General Catalogue and Index Catalogue (RNGC/IC) \citep{galaxydata}, which maps the sky positions of 10,608 galaxies. This data set has been analysed extensively in previous methodological works for point patterns on a sphere \citep{Lawrence2016,Cuevas-Pacheco2018}, but only from a univariate perspective. In these papers, additional descriptions of each galaxy were ignored to create a single-type point pattern. Here,we take into account marks that label the type (shape) of each galaxy. These marks are hierarchical, and as such, we consider just the first level of that hierarchy which indicates the broad shape of each galaxy. We further restrict ourselves to a bivariate dataset of the two most common galaxy types, spiral (6,674 galaxies) and elliptical (1,231 galaxies), as shown in Figure \ref{fig:galaxies}. We follow \cite{Lawrence2016} and \cite{Cuevas-Pacheco2018} in removing the Milky Way band, in which galaxies are mostly obscured. Furthermore, the data is rotated such that this band is orientated along the equator.


It has been noted in \cite{Cuevas-Pacheco2018} that this point pattern is inhomogeneous. Instead of adopting their approach of applying a thinning procedure to homogenise the data, we treat it as an inhomogeneous point pattern and make use the intensity estimator of \cite{Ward2023} to plug into the inhomogeneous versions of the functional summary statistics, as presented in Section \ref{inhomo_sphere}.

\begin{figure}
	\centering
	\begin{minipage}{\textwidth}
		\centering
		\subcaptionbox{The RNGC galaxy point pattern observed from two different viewpoints. Spiral galaxies are shown in cyan and elliptical galaxies are shown in magenta. The point pattern has been rotated to give an equatorial orientation to the milky way, and events within the milky way band (shown with the thick black lines at $\pm 12\deg$ latitude) have been removed. \label{fig:galaxies}}{
			\includegraphics[width=0.9\textwidth]{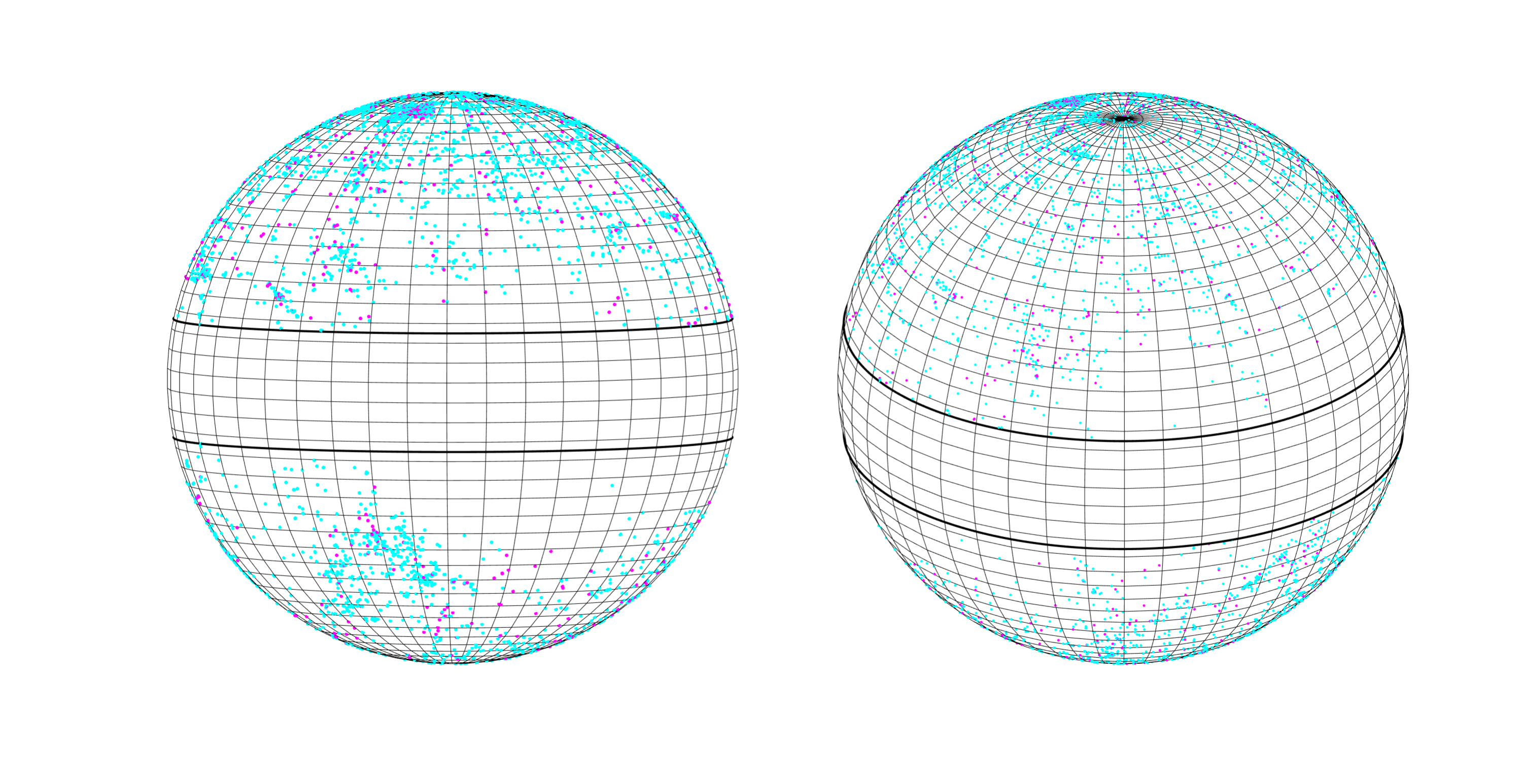}
		}
	\end{minipage}

	\begin{minipage}{\textwidth}
		\centering
		\subcaptionbox{Estimated functional summary statistics $\hat P^{12}(r)$, $\hat J^{12}(r)$ and $\hat J^{21}(r)$ for the RNGC galaxy point pattern. Grey bands depict the 95\% simulation envelopes for the null distribution using the Poisson simulation method. The theoretical value under the null is shown with the red dotted line.\label{fig:FSS_galaxy}}{
			\includegraphics[width=0.95\textwidth]{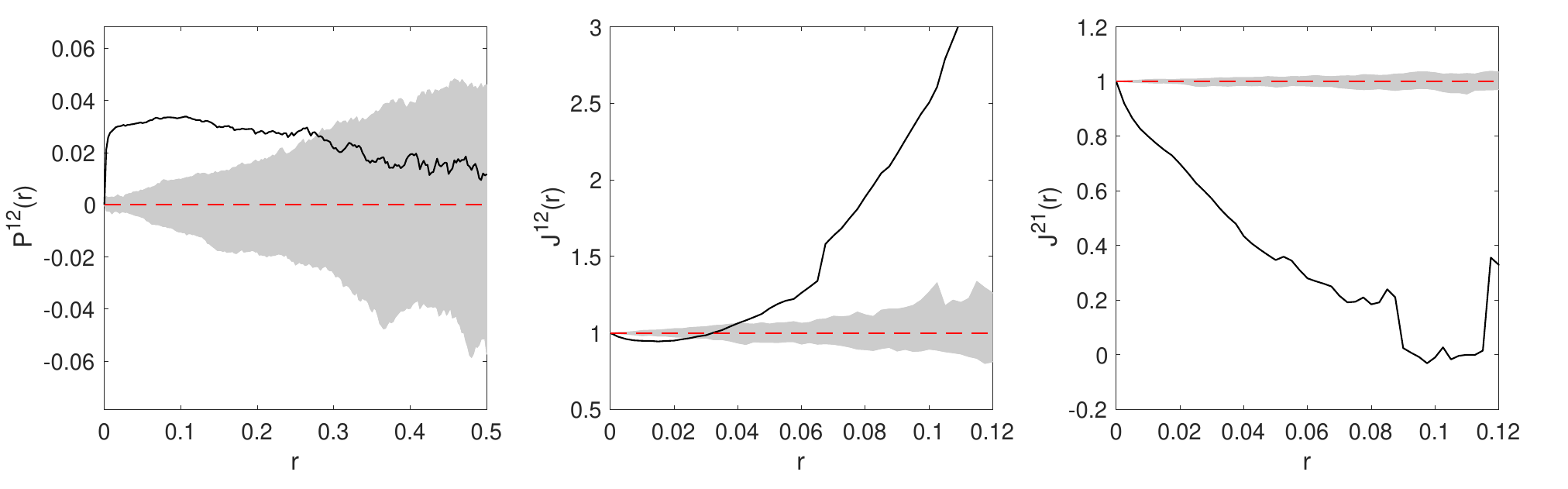}
		}
	\end{minipage}

	\caption{RNGC galaxy point pattern analysis.}
	\label{fig:main}
\end{figure}

Figure \ref{fig:FSS_galaxy} plots estimates of the functional summary statistics $P^{12}(\cdot)$, $J^{12}(\cdot)$ and $J^{21}(\cdot)$. Simulation intervals for the null hypothesis of independence were computed using 199 simulates of two independent Poisson processes, each with intensity functions as estimated for the spiral and elliptical galaxy point patterns, respectively.

As shown, $\hat P^{12}(\cdot)$ lies significantly above the simulation intervals, and $\hat J^{21}(\cdot)$ lies significantly below the simulation envelopes. These results are consistent with each other, providing compelling evidence that spiral and elliptical galaxies show attractive dependence, even when accounting for inhomogeneity in the individual processes. Note that $\hat J^{12}(\cdot)$ also displays behaviour at small $r$ that is consistent with this inference, however the reason why it goes above the simulation interval for larger distances is unexplained and motivates further analysis.



\section{Discussion}
To aid the exposition, we have elected to restrict the presented work to the surface of convex shapes. It is worth noting that there is a straight-forward extension to a class of non-convex shapes which we call \emph{star-shapes}. The set $\mathbb{D}_{int}$ is said to be a star domain if and only if there exists an $\mathbf{x}_0\in \mathbb{D}_{int}$ such that for any $\mathbf{x}\in \mathbb{D}_{int}$, the line-segment $\{\mathbf{z}\in\mathbb{R}^3:\mathbf{z}=\mathbf{x}_0+\gamma(\mathbf{x}-\mathbf{x}_0),\; \gamma\in (0,1)\}\subset \mathbb{D}_{int}$. We say $\mathbb{D}_{int}$ is star-shaped at $\mathbf{x}_0$. We further define $\mathbb{D}$ to be star-shaped if and only if its interior $\mathbb{D}_{int}$ is a star domain. Note that all convex shapes are star domains.

If $\mathbb{D}_{int}$ is star-shaped at $\mathbf{x}_0$, then $\mathbb{D}_{int}-\mathbf{x}_0$ is star-shaped at $\mathbf{0}$. We can therefore assume, without loss of generality, that a star domain is star-shaped at $\mathbf{0}\in \mathbb{D}_{int}$. In this setting $f(\mathbf{x}) = \mathbf{x}/\|\mathbf{x}\|$ is a bijective mapping from $\mathbb{D}$ to $\mathbb{S}^2$, and Theorem \ref{thm:mapping} still holds. Analysis can then proceed in the usual way. The complication of working with star-shapes over convex shapes comes form formulating mathematical representations for $\mathbb{D}$ and simulating point patterns on it.

\section*{Acknowledgements}
Scott Ward is supported by Wellcome Trust grant 203799/Z/16/Z.

\newpage

\appendix

\section{Invariance of $K^{CE}$ to the typical point}
\label{Kinvariance}
Representing $K^{CE}(\cdot)$ as an expectation of the reduced Palm distribution of $X$ using the Campbell-Mecke formula,
\begin{align*}
K^{CE}(r) &= \frac{1}{\lambda_{\mathbb{S}^2}(A)\nu(C)\nu(E)}\mathbb{E}\left[\mathop{\sum\nolimits\sp{\ne}}_{(\mathbf{x},m_{\mathbf{x}}),(\mathbf{y},m_{\mathbf{y}})\in X}\frac{\mathbbm{1}\left\{(\mathbf{x},m_\mathbf{x})\in A\times C, (O^T_{\mathbf{x}}\mathbf{y},m_\mathbf{y})\in B_{\mathbb{S}^2}(\mathbf{o},r)\times E\right\}}{\rho(\mathbf{x},m_{\mathbf{x}}) \rho(\mathbf{y},m_{\mathbf{y}})}\right]\\
&= \frac{1}{\lambda_{\mathbb{S}^2}(A)\nu(C)\nu(E)}\mathbb{E}\left[\mathop{\sum\nolimits\sp{\ne}}_{(\mathbf{x},m_{\mathbf{x}}),(\mathbf{y},m_{\mathbf{y}})\in X}\frac{\mathbbm{1}\left\{(\mathbf{x},m_\mathbf{x})\in A\times C, (\mathbf{y},m_\mathbf{y})\in B_{\mathbb{S}^2}(\mathbf{x},r)\times E\right\}}{\rho(\mathbf{x},m_{\mathbf{x}}) \rho(\mathbf{y},m_{\mathbf{y}})}\right]\\
&=\frac{1}{\lambda_{\mathbb{S}^2}(A)\nu(C)\nu(E)}\int_{A}\int_{C}\mathbb{E}^{!}_{(\mathbf{x},m)}\left[\mathop{\sum\nolimits}_{(\mathbf{y},m_{\mathbf{y}})\in X}\frac{\mathbbm{1}\left\{(\mathbf{y},m_\mathbf{y})\in B_{\mathbb{S}^2}(\mathbf{x},r)\times E\right\}}{\rho(\mathbf{x},m_{\mathbf{x}}) \rho(\mathbf{y},m_{\mathbf{y}})}\right]\rho(\mathbf{x},m)d\nu(m)d\lambda_{\mathbb{S}^2}(\mathbf{x}).
\end{align*}
By definition, $K^{CE}(\cdot)$ is independent of $A$ and hence the integrand must be constant with respect to $\mathbf{x}$. Therefore, $\mathbf{x}$ can be taken as any point in $\mathbb{S}^2$. In particular we can set it to $\mathbf{o}$,
\begin{align*}
K^{CE}(r)&=\frac{1}{\nu(C)\nu(E)}\int_{C}\mathbb{E}^{!}_{(\mathbf{o},m)}\left[\mathop{\sum\nolimits}_{(\mathbf{y},m_{\mathbf{y}})\in X}\frac{\mathbbm{1}\left\{(\mathbf{y},m_\mathbf{y})\in B_{\mathbb{S}^2}(\mathbf{o},r)\times E\right\}}{ \rho(\mathbf{y},m_{\mathbf{y}})}\right]d\nu(m)\\
&=\frac{1}{\nu(E)}\mathbb{E}^{!}_{\mathbf{o},C}\left[\mathop{\sum\nolimits}_{(\mathbf{y},m_{\mathbf{y}})\in X}\frac{\mathbbm{1}\left\{(\mathbf{y},m_\mathbf{y})\in B_{\mathbb{S}^2}(\mathbf{o},r)\times E\right\}}{ \rho(\mathbf{y},m_{\mathbf{y}})}\right].
\end{align*}

\section{Proof of Proposition \ref{prop:iso:function}}\label{proof:iso:function}
\begin{proof}
	This proof follows that of \cite[Proposition 4.5]{Moller2004}. Clearly $J^{CE}(r)=1$ if $D^{CE}(r)=F^{E}(r)$. We can rewrite $D^{CE}$ as an expectation over the point process $X$ rather then its reduced Palm process $X^{!}_{\mathbf{o},C}$. Let our finite reference measure $\nu$ over $\mathcal{M}$ coincide with the mark distribution and so under isotropy we have that $\rho(\mathbf{x},m)=\rho_g$ and $A$ be an arbitrary subset of $\mathbb{S}^2$ such that $\lambda_{\mathbb{S}^2}(A)>0$ then,
	\begin{align}
	D^{CE}(r) &= P[X^!_{\mathbf{o},C}\cap\{B_{\mathbb{S}^2}(\mathbf{o},r)\times E\}\neq \emptyset]\nonumber\\
	&= P[X^!_{\mathbf{x},C}\cap\{B_{\mathbb{S}^2}(\mathbf{x},r)\times E\}\neq \emptyset]\label{eq:appendix:A:1}\\
	&=\frac{1}{\lambda_{\mathbb{S}^2}(A)}\int_A P[X^!_{\mathbf{x},C}\cap\{B_{\mathbb{S}^2}(\mathbf{x},r)\times E\}\neq \emptyset] \lambda_{\mathbb{S}^2}(d\mathbf{x})\nonumber\\
	&=\frac{1}{\lambda_{\mathbb{S}^2}(A)\nu(C)}\int_A\int_C P[X^!_{\mathbf{x},C}\cap\{B_{\mathbb{S}^2}(\mathbf{x},r)\times E\}\neq \emptyset] \lambda_{\mathbb{S}^2}(d\mathbf{x})\nu(dm)\label{eq:appendix:A:2}\\
	&=\frac{1}{\rho\lambda_{\mathbb{S}^2}(A)\nu(C)}\int_A\int_C P[X^!_{\mathbf{x},C}\cap\{B_{\mathbb{S}^2}(\mathbf{x},r)\times E\}\neq \emptyset]\rho \lambda_{\mathbb{S}^2}(d\mathbf{x})\nu(dm)\nonumber\\
	&=\frac{1}{\rho\lambda_{\mathbb{S}^2}(A)\nu(C)}\mathbb{E}\left(\sum_{(\mathbf{x},m)\in X}\mathbbm{1}[(\mathbf{x},m)\in A\times C, \{X\setminus (\mathbf{x},m)\}\cap\{B_{\mathbb{S}^2}(\mathbf{x},r)\times E\}\neq \emptyset]\right),\label{eq:appendix:A:3}
	\end{align}
	where (\ref{eq:appendix:A:1}) follows by isotropy of $X$, (\ref{eq:appendix:A:2}) follows by definition of $P^!_{\mathbf{x},C}$ and (\ref{eq:appendix:A:3}) by the Campbell-Mecke theorem. Let us define $X_{A\times C}= X\cap (A\times C)$ for $A\subset\mathbb{S}^2$ and $C\subset \mathcal{M}$. Noting that $C\cap E = \emptyset$, $(\mathbf{x},m)\in X_{A\cap C}$ implies that the event $\{X\setminus (\mathbf{x},m)\}\cap\{B_{\mathbb{S}^2}(\mathbf{x},r)\times E\}\neq \emptyset$ is equivalent to $X\cap\{B_{\mathbb{S}^2}(\mathbf{x},r)\times E\}\neq \emptyset$. This expectation can therefore be rewritten as,
	\begin{align}	D^{CE}(r)&=\mathbb{E}\left(\sum_{(\mathbf{x},m)\in X_{A\times C}}\mathbb{E}\left[\mathbbm{1}\{X_{B_{\mathbb{S}^2}(\mathbf{x},r)\times E}\neq \emptyset\}|X_C\right]\right)\label{eq:appendix:a:4}\\
	&= \frac{1}{\rho\lambda_{\mathbb{S}^2}(A)\nu(C)}\mathbb{E}\left[\sum_{(\mathbf{x},m)\in X_{A\times C}}P\{X_{B_{\mathbb{S}^2}(\mathbf{x},r)\times E}\neq \emptyset\}\right]\label{eq:appendix:a:5}\\
	&= \frac{1}{\rho\lambda_{\mathbb{S}^2}(A)\nu(C)}\mathbb{E}\left\{\sum_{(\mathbf{x},m)\in X_{A\times C}}F^{E}(r)\right\}\nonumber\\
	&= \frac{\rho\lambda_{\mathbb{S}^2}(A)\nu(C)}{\rho\lambda_{\mathbb{S}^2}(A)\nu(C)}F^{E}(r)\nonumber\\
	& = F^{E}(r)\nonumber,
	\end{align}
	where (\ref{eq:appendix:a:4}) follows from the law of iterated expectation and (\ref{eq:appendix:a:5}) follows from independence of $X_C$ and $X_E$.

	To show that $K^{CE}(r)=2\pi\{1-\cos(r)\}$ we note that under independence $\rho^{(2)}\{(\mathbf{x},m_{\mathbf{x}}),(\mathbf{y},m_{\mathbf{y}})\}=\rho(\mathbf{x},m_{\mathbf{x}})\rho(\mathbf{y},m_{\mathbf{y}})$ when $m_{\mathbf{x}}\in C$ and $m_{\mathbf{y}}\in E$. The result follows by application of the Campbell theorem and noting that $\lambda_{\mathbb{S}^2}\{B_{\mathbb{S}^2}(\mathbf{o},r)\}=2\pi\{1-\cos(r)\}$.
\end{proof}

\section{Proof of Proposition \ref{prop:iso:bias}}\label{proof:iso:bias}
\begin{proof}
	Starting with $\hat{F}^{E}(\cdot)$ we have that,
	\begin{align*}
	\mathbb{E}\{1-\hat{F}^{E}(r)\} &= \frac{1}{|I_{W_{\ominus r}}|}\sum_{\mathbf{p}\in I_{W_{\ominus r}}}\mathbb{E}\left(\prod_{(\mathbf{x},m)\in X} \left[1-\mathbbm{1}\{d_{\mathbb{S}^2}(\mathbf{p},\mathbf{x})\leq r, m\in E\}\right]\right)\\
	&=\frac{1}{|I_{W_{\ominus r}}|}\sum_{\mathbf{p}\in I_{W_{\ominus r}}} P[X\cap \{B_{\mathbb{S}^2}(\mathbf{p},r)\times E\} = \emptyset]\\
	&=\frac{1}{|I_{W_{\ominus r}}|}\sum_{\mathbf{p}\in I_{W_{\ominus r}}} P[X\cap \{B_{\mathbb{S}^2}(\mathbf{o},r)\times E\} = \emptyset]\\
	&=P[X\cap \{B_{\mathbb{S}^2}(\mathbf{o},r)\times E\} = \emptyset]\\
	&=1-F^{E}(r).
	\end{align*}
	For $\hat{D}^{CE}(\cdot)$, by the Campbell-Mecke Theorem,
	\begin{align*}
	\mathbb{E}\{1-\hat{D}^{CE}(r)\}&=\mathbb{E}\left\{\frac{1}{\rho_g \lambda_{\mathbb{S}^2}(W_{\ominus r})\nu(C)}\right.\\
	&\qquad\left. \sum_{(\mathbf{x},m_{\mathbf{x}})\in X_{W\ominus r}} \mathbbm{1}\{(\mathbf{x},m_{\mathbf{x}})\in W_{\ominus r}\times C\}\prod_{(\mathbf{y},m_{\mathbf{y}})\in X}\left[1- \mathbbm{1}\{d_{\mathbb{S}^2}(\mathbf{x},\mathbf{y}) < r,m_{\mathbf{y}}\in E\}\right]\right\}\\
	&=\frac{1}{\rho_g \lambda_{\mathbb{S}^2}(W_{\ominus r})\nu(C)}\int_{W_{\ominus r}}\int_C\\
	&\phantom{AAAA}\mathbb{E}^!_{(\mathbf{x},m)}\left(\prod_{(\mathbf{y},m_{\mathbf{y}})\in X}\left[1- \mathbbm{1}\{d_{\mathbb{S}^2}(\mathbf{x},\mathbf{y}) < r,m_{\mathbf{y}}\in E\} \right] \rho(\mathbf{x},m)\lambda_{\mathbb{S}^2}(d\mathbf{x})\nu(m)\right).\\
	\intertext{Since $X$ is isotropic, and we have taken the reference measure $\nu$ to be the probability measure over $\mathcal{M}$, we have that $\rho(\mathbf{x},m)=\rho_g$ and so}
	&=\frac{1}{\lambda_{\mathbb{S}^2}(W_{\ominus r})\nu(C)}\int_{W_{\ominus r}}\int_C\mathbb{E}^!_{(\mathbf{x},m)}\left\{\prod_{(\mathbf{y},m_{\mathbf{y}})\in X}\left[1- \mathbbm{1}\{d_{\mathbb{S}^2}(\mathbf{x},\mathbf{y}) < r,m_{\mathbf{y}}\in E\}\right] \right\} \lambda_{\mathbb{S}^2}(d\mathbf{x})\nu(m)\\
	&=\frac{1}{\lambda_{\mathbb{S}^2}(W_{\ominus r})}\int_{W_{\ominus r}}\mathbb{E}^!_{\mathbf{x},C}\left\{\prod_{(\mathbf{y},m_{\mathbf{y}})\in X}\left[1- \mathbbm{1}\{d_{\mathbb{S}^2}(\mathbf{x},\mathbf{y}) < r,m_{\mathbf{y}}\in E\}\right] \right\} \lambda_{\mathbb{S}^2}(d\mathbf{x})\\
	&=\frac{1}{\lambda_{\mathbb{S}^2}(W_{\ominus r})}\int_{W_{\ominus r}}P[X^!_{\mathbf{x},C}\cap\{B_{\mathbb{S}^2}(\mathbf{x},r)\times E\}= \emptyset]\lambda_{\mathbb{S}^2}(d\mathbf{x})\\
	&=\frac{1}{\lambda_{\mathbb{S}^2}(W_{\ominus r})}\int_{W_{\ominus r}}1-P[X^!_{\mathbf{x},C}\cap\{B_{\mathbb{S}^2}(\mathbf{x},r)\times E\}\neq \emptyset]\lambda_{\mathbb{S}^2}(d\mathbf{x})\\
	&=\frac{1}{\lambda_{\mathbb{S}^2}(W_{\ominus r})}\int_{W_{\ominus r}}1-P[X^!_{\mathbf{o},C}\cap\{B_{\mathbb{S}^2}(\mathbf{x},r)\times E\}\neq \emptyset]\lambda_{\mathbb{S}^2}(d\mathbf{x})\\
	&=1-P[X^!_{\mathbf{o},C}\cap\{B_{\mathbb{S}^2}(\mathbf{x},r)\times E\}\neq \emptyset]\\
	&=1-D^{CE}(r).
	\end{align*}
	For $\hat{K}^{CE}(\cdot)$, unbiasedness follows by application of the Campbell Theorem and the fact the independence between $X_C$ and $X_E$ implies $\rho^{(2)}\{(\mathbf{x},m_{\mathbf{x}}),(\mathbf{y},m_{\mathbf{y}})\}=\rho(\mathbf{x},m_{\mathbf{x}})\rho(\mathbf{y},m_{\mathbf{y}})$ when $m_1\in C$ and $m_2\in E$. Ratio-unbiasedness of $\hat{J}^{CE}(\cdot)$ follows by unbiasedness of $\hat{F}^{E}(\cdot)$ and $\hat{D}^{CE}(\cdot)$.
\end{proof}

\section{Proof of Proposition \ref{prop:equiv:iso:inhom}}\label{proof:equiv:iso:inhom}
\begin{proof}
	Starting with $D^{CE}_{\rm inhom}(\cdot)$, suppose that $X$ is now isotropic and that the reference measure $\nu$ takes the canonical mark measure, then $D^{CE}_{\rm inhom}(\cdot)$ becomes,
	\begin{align*}
	D^{CE}_{\rm inhom}(r) &= 1 - G^!_{\mathbf{y},C}(1-u^r_{\mathbf{y},E})\\
	&= 1 - \frac{1}{\nu(C)}\int_C \mathbb{E}^!_{(\mathbf{y},m)}\left(\prod_{(\mathbf{x},n)\in X}\left[1-\frac{\bar{\rho}_E\mathbbm{1}\{(\mathbf{x},n)\in B_{\mathbb{S}^2}(\mathbf{y},r)\times E\}}{\rho(\mathbf{x},n)}\right]\right)\nu(dm)\\
	\intertext{since the process is isotopic $\rho(\mathbf{x},m)\equiv\rho\in\mathbb{R}_+$ is a constant and $\bar{\rho}_E=\rho$}
	&= 1 - \frac{1}{\nu(C)}\int_C \mathbb{E}^!_{(\mathbf{y},m)}\left(\prod_{(\mathbf{x},n)\in X}\left[1-\mathbbm{1}\{(\mathbf{x},n)\in B_{\mathbb{S}^2}(\mathbf{y},r)\times E\}\right]\right)\nu(dm)\\
	&= 1 - \frac{1}{\nu(C)}\int_C \mathbb{E}^!_{(\mathbf{y},m)}\left[\prod_{(\mathbf{x},n)\in X} \mathbbm{1}\{(\mathbf{x},n)\notin B_{\mathbb{S}^2}(\mathbf{y},r)\times E\}\right]\nu(dm)\\
	&= 1 - \frac{1}{\nu(C)}\int_C \mathbb{E}^!_{(\mathbf{y},m)}\left( \mathbbm{1}[X\cap\{ B_{\mathbb{S}^2}(\mathbf{y},r)\times E\}=\emptyset]\right)\nu(dm)\\
	&= 1 - \frac{1}{\nu(C)}\int_C \mathbb{P}^!_{(\mathbf{y},m)}\left[X\cap\{ B_{\mathbb{S}^2}(\mathbf{y},r)\times E\}=\emptyset\right]\nu(dm)\\
	&= 1 - \mathbb{P}^!_{\mathbf{y},C}\left[X\cap\{ B_{\mathbb{S}^2}(\mathbf{y},r)\times E\}=\emptyset\right]\\
	&= \mathbb{P}^!_{\mathbf{y},C}\left[X\cap\{ B_{\mathbb{S}^2}(\mathbf{y},r)\times E\}\neq\emptyset\right].\\
	\intertext{Since $X$ is isotropic, this probability does not depend on $\mathbf{y}$ and we have}
	D^{CE}_{\rm inhom}(r)&= \mathbb{P}^!_{\mathbf{o},C}\left(X\cap( B_{\mathbb{S}^2}(\mathbf{o},r)\times E)\neq\emptyset\right)
	\end{align*}
	which is precisely (\ref{eq:iso:sphere:D}).

	Next consider $F^{E}_{\rm inhom}(\cdot)$. To see why this is an extension of the isotropic $F$-function given by (\ref{eq:iso:sphere:F}), consider $X$ to be an istropic spheroidal mark process. Then,
	\begin{align*}
	F^{E}_{\rm inhom}(r) &= 1 - G(1-u^r_{\mathbf{y},E})\\
	&= 1 - \mathbb{E}\left(\prod_{(\mathbf{x},m)\in X}\left[1-\frac{\bar{\rho}_E\mathbbm{1}\{(\mathbf{x},m)\in B_{\mathbb{S}^2}(\mathbf{y},r)\times E\}}{\rho(\mathbf{x},m)}\right]\right).\\
	\intertext{Since the process is isotopic, $\rho(\mathbf{x},m)\equiv\rho\in\mathbb{R}_+$ is a constant and $\bar{\rho}_E=\rho$. Therefore}
	F^{E}_{\rm inhom}(r)	&= 1 - \mathbb{E}\left(\prod_{(\mathbf{x},m)\in X}\left[1-\mathbbm{1}\{(\mathbf{x},m)\in B_{\mathbb{S}^2}(\mathbf{y},r)\times E\}\right]\right)\\
	&= 1 - \mathbb{E}\left[\prod_{(\mathbf{x},m)\in X} \mathbbm{1}\{(\mathbf{x},m)\notin B_{\mathbb{S}^2}(\mathbf{y},r)\times E\}\right]\\
	&= 1 - \mathbb{E}\left(\mathbbm{1}[X\cap \{B_{\mathbb{S}^2}(\mathbf{y},r)\times E\}=\emptyset]\right)\\
	&= 1 - \mathbb{P}(X\cap (B_{\mathbb{S}^2}(\mathbf{y},r)\times E)=\emptyset)\\
	&= \mathbb{P}[X\cap \{B_{\mathbb{S}^2}(\mathbf{y},r)\times E\}\neq\emptyset].
	\intertext{Since $X$ is isotropic, this probability does not depend on $\mathbf{y}$, and we have}
	&= \mathbb{P}[X\cap \{B_{\mathbb{S}^2}(\mathbf{o},r)\times E\}\neq\emptyset],
	\end{align*}
	which is precisely (\ref{eq:iso:sphere:F}).

\end{proof}

\section{Proof of Theorem \ref{thm:J:inhom:generating}}\label{proof:J:inhom:generating}

Before providing the proof we state and prove the following lemma, similar to a portion of the proof to Theorem 1 by \cite{Cronie2016}.
\begin{lemma}\label{lemma:app:1}
	Let $X$ be a marked point process on $\mathbb{S}^2$ that is IRWMI. Then,
	\begin{align*}
	\int_B&\left\{\int_C\mathbb{E}^{!}_{(\mathbf{y},a)}\left[{\sum\nolimits\sp{\ne}_{(\mathbf{x}_1,m_1),\dots,(\mathbf{x}_1,m_1)\in X}}\prod_{i=1}^n\frac{\mathbbm{1}\{(\mathbf{x}_i,m_i)\in B_{\mathbb{S}^2}(\mathbf{y},r)\times E\}}{\rho(\mathbf{x}_i,m_i)} \right]\nu(da)\right\}\lambda(d\mathbf{y})\\
	&=\int_B\left\{\int_C\left[\int_{\{B_{\mathbb{S}^2}(\mathbf{o},r)\times E\}^n}\frac{\rho^{(n+1)}\{(\mathbf{o},a),(\mathbf{x}_1,m_1),\dots,(\mathbf{x}_n,m_n)\}}{\rho(\mathbf{o},a)\rho(\mathbf{x}_1,m_1)\cdots\rho(\mathbf{x}_n,m_n)}\prod_{i=1}^n\lambda_{\mathbb{S}^2}(d\mathbf{x}_i)\nu(dm_i)\right]\nu(da)\right\}\lambda(d\mathbf{y}),
	\end{align*}
	and so the integrand in the curly brackets are almost $\lambda_{\mathbb{S}^2}$ everywhere equal, and the integrand on the left hand side is independent of $\mathbf{y}\in\mathbb{S}^2$.
\end{lemma}
\begin{proof}
	Define the following function $h:(\mathbb{S}^2\times \mathcal{M}) \times N_{lf}\mapsto \mathbb{R}_+$,
	\begin{align*}
	h_{r,B}\{(\mathbf{y},a), \phi\} = \frac{\mathbbm{1}\{(\mathbf{y},a)\in B\times C\}}{\rho(\mathbf{y},a)}{\sum\nolimits\sp{\ne}_{(\mathbf{x}_1,m_1),\dots,(\mathbf{x}_1,m_1)\in \phi}}\prod_{i=1}^n\frac{\mathbbm{1}\{(\mathbf{x}_i,m_i)\in B_{\mathbb{S}^2}(\mathbf{y},r)\times E\}}{\rho(\mathbf{x}_i,m_i)}.
	\end{align*}
	Then by taking expectations of $h_{r,B}$ we have,
	\begin{align*}
	\mathbb{E}&\left[\sum_{(\mathbf{y},a)\in X} h_{r,B}\{(\mathbf{y},a), X\setminus (\mathbf{y},a)\}\right]\\
	&=\mathbb{E}\left[\sum_{(\mathbf{y},a)\in X}\frac{\mathbbm{1}\{(\mathbf{y},a)\in B\times C\}}{\rho(\mathbf{y},a)}{\sum\nolimits\sp{\ne}_{(\mathbf{x}_1,m_1),\dots,(\mathbf{x}_1,m_1)\in X\setminus (\mathbf{y},a)}}\prod_{i=1}^n\frac{\mathbbm{1}\{(\mathbf{x}_i,m_i)\in B_{\mathbb{S}^2}(\mathbf{y},r)\times E\}}{\rho(\mathbf{x}_i,m_i)}\right]\\
	&=\mathbb{E}\left[{\sum\nolimits\sp{\ne}_{(\mathbf{y},a),(\mathbf{x}_1,m_1),\dots,(\mathbf{x}_1,m_1)\in X}} \frac{\mathbbm{1}\{(\mathbf{y},a)\in B\times C\}}{\rho(\mathbf{y},a)}\prod_{i=1}^n\frac{\mathbbm{1}\{(\mathbf{x}_i,m_i)\in B_{\mathbb{S}^2}(\mathbf{y},r)\times E\}}{\rho(\mathbf{x}_i,m_i)}\right]\\
	&=\int_B\int_C\left\{\int_{(B_{\mathbb{S}^2}(\mathbf{y},r)\times E)^n} \frac{\rho^{(n+1)}\{(\mathbf{y},a),(\mathbf{x}_1,m_1),\dots,(\mathbf{x}_n,m_n)\}}{\rho(\mathbf{y},a)\rho(\mathbf{x}_1,m_1)\cdots\rho(\mathbf{x}_n,m_n)} \prod_{i=1}^n \lambda_{\mathbb{S}^2}(d\mathbf{x}_i)\nu(dm_i) \right\}\lambda_{\mathbb{S}^2}(d\mathbf{y})\nu(da).\\
	\intertext{By the IRWMI assumption, we can remove the dependency on $\mathbf{y}$ to give,}
	&=\int_B\int_C\left\{\int_{(B_{\mathbb{S}^2}(\mathbf{o},r)\times E)^n} \frac{\rho^{(n+1)}\{(\mathbf{o},a),(\mathbf{x}_1,m_1),\dots,(\mathbf{x}_n,m_n)\}}{\rho(\mathbf{o},a)\rho(\mathbf{x}_1,m_1)\cdots\rho(\mathbf{x}_n,m_n)} \prod_{i=1}^n \lambda_{\mathbb{S}^2}(d\mathbf{x}_i)\nu(dm_i) \right\}\lambda_{\mathbb{S}^2}(d\mathbf{y})\nu(da).
	\end{align*}
	We can also use the Campbell-Mecke Theorem to obtain,
	\begin{align*}
	\mathbb{E}&\left[\sum_{(\mathbf{y},a)\in X} h_{r,B}\{(\mathbf{y},a), X\setminus (\mathbf{y},a)\}\right]\\
	&=\int_B\left\{\int_C\mathbb{E}^{!}_{(\mathbf{y},a)}\left[{\sum\nolimits\sp{\ne}_{(\mathbf{x}_1,m_1),\dots,(\mathbf{x}_1,m_1)\in X}}\prod_{i=1}^n\frac{\mathbbm{1}\{(\mathbf{x}_i,m_i)\in B_{\mathbb{S}^2}(\mathbf{y},r)\times E\}}{\rho(\mathbf{x}_i,m_i)} \right]\nu(da)\right\}\lambda(d\mathbf{y}),
	\end{align*}
	and hence the lemma holds.
\end{proof}

We now commence with the proof of Theorem \ref{thm:J:inhom:generating}
\begin{proof}
	Our proof is analagous to that of \cite[Theorem 1][]{Cronie2016} where instead of considering $\mathbb{R}^n$ we consider $\mathbb{S}^2$. We can consider an infinite series expansion of $F^{E}_{\text{inhom}}(\cdot)$ by using the infinite series representation of the generating functional,
	\begin{align*}
	1-F^{E}_{\text{inhom}}(r) &= G(1-u^r_{\mathbf{y},E})\\
	&= 1+\sum_{n=1}^{\infty}\frac{(-1)^n}{n!}\int_{\mathbb{S}^2\times\mathcal{M}}\cdots\int_{\mathbb{S}^2\times\mathcal{M}}\prod_{i=1}^n \frac{\bar{\rho}_E\mathbbm{1}\{(\mathbf{x}_i,m_i)\in B_{\mathbb{S}^2}(\mathbf{y},r)\times E\}}{\rho(\mathbf{x}_i,m_i)}\\
	&\rho^{(n)}\left\{(\mathbf{x}_1,m_1),\dots,(\mathbf{x}_n,m_n)\right\}\lambda_{\mathbb{S}^2}(d\mathbf{x}_i)\nu(m_i),\\
	&= 1+\sum_{n=1}^{\infty}\frac{(-\bar{\rho}_E)^n}{n!}\int_{B_{\mathbb{S}^2}(\mathbf{y},r)\times E}\cdots\int_{B_{\mathbb{S}^2}(\mathbf{y},r)\times E} \frac{ \rho^{(n)}\left\{(\mathbf{x}_1,m_1),\dots,(\mathbf{x}_n,m_n)\right\}}{\rho(\mathbf{x}_1,m)\cdots\rho(\mathbf{x}_n,m)}\prod_{i=1}^n\lambda_{\mathbb{S}^2}(d\mathbf{x}_i)\nu(m_i).\\
	\intertext{Using e.g. \cite[p. 126]{chiu2013stochastic}, the above is equivalent to}
	&= \exp\left\{\sum_{n=1}^{\infty}\frac{(-\bar{\rho}_E)^n}{n!}\int_{B_{\mathbb{S}^2}(\mathbf{y},r)\times E}\cdots\int_{B_{\mathbb{S}^2}(\mathbf{y},r)\times E}\xi_{n}((\mathbf{x}_1,m_1),\dots,(\mathbf{x}_n,m_n))\prod_{i=1}^n\lambda_{\mathbb{S}^2}(d\mathbf{x}_i)\nu(m_i)\right\}\\
	&= \exp\left\{\sum_{n=1}^{\infty}\frac{(-\bar{\rho}_E)^n}{n!}\int_{B_{\mathbb{S}^2}(\mathbf{o},r)\times E}\cdots\int_{B_{\mathbb{S}^2}(\mathbf{o},r)\times E}\xi_{n}((\mathbf{x}_1,m_1),\dots,(\mathbf{x}_n,m_n))\prod_{i=1}^n\lambda_{\mathbb{S}^2}(d\mathbf{x}_i)\nu(m_i)\right\},
	\end{align*}
	where the final equality is by IRWMI of $X$. Hence $F^E_{\text{inhom}}(\cdot)$ is independent of the arbitrary point $\mathbf{y}$.

Next we look at $D^{CE}_{\text{inhom}}(\cdot)$,
	\begin{align*}
	1-D^{CE}_{\text{inhom}}(r) &= G^!_{\mathbf{y},C}(1-u^r_{\mathbf{y},E})\\
	&=\frac{1}{\nu(C)}\int_C \mathbb{E}^{!}_{(\mathbf{x},a)}\left[\prod_{(\mathbf{y},m)\in X}\left\{1-\frac{\bar{\rho}_E\mathbbm{1}\{(\mathbf{x}_i,m_i)\in B_{\mathbb{S}^2}(\mathbf{y},r)\times E\}}{\rho(\mathbf{x}_i,m_i)}\right\}\right]\nu(da)\\
	&=1+\frac{1}{\nu(C)}\sum_{n=1}^{\infty}\frac{(-\bar{\rho}_E)^n}{n!}\\
	&\int_C\mathbb{E}^{!}_{(\mathbf{x},a)}\left[{\sum\nolimits\sp{\ne}_{(\mathbf{x}_1,m_1),\dots,(\mathbf{x}_1,m_1)\in X}}\prod_{i=1}^n\frac{\mathbbm{1}\{(\mathbf{x}_i,m_i)\in B_{\mathbb{S}^2}(\mathbf{y},r)\times E\}}{\rho(\mathbf{x}_i,m_i)} \right]\nu(da),\\
	\intertext{where the final equality holds due to local finitedness of $X$ and using an inclusion-exclusion argument. From Lemma \ref{lemma:app:1} we have,}
	1-D^{CE}_{\text{inhom}}(r)&=1+\frac{1}{\nu(C)}\sum_{n=1}^{\infty}\frac{(-\bar{\rho}_E)^n}{n!}\\
	&\int_C\int_{\{B_{\mathbb{S}^2}(\mathbf{o},r)\times E\}^n}\frac{\rho^{(n+1)}\{(\mathbf{o},a),(\mathbf{x}_1,m_1),\dots,(\mathbf{x}_n,m_n)\}}{\rho(\mathbf{o},a)\rho(\mathbf{x}_1,m_1)\cdots\rho(\mathbf{x}_n,m_n)}\prod_{i=1}^n\lambda_{\mathbb{S}^2}(d\mathbf{x}_i)\nu(dm_i)\nu(da)\\
	&=1+\frac{1}{\nu(C)}\sum_{n=1}^{\infty}\frac{(-\bar{\rho}_E)^n}{n!}\int_C\\
	&\left(\int_{\{B_{\mathbb{S}^2}(\mathbf{o},r)\times E\}^n}\sum_{k=1}^{n+1}\sum_{E_1,\dots,E_k}\prod_{j=1}^{k}\xi_{|E_j|}[\{(\mathbf{x}_i,m_i):i\in E_j\}]\prod_{i=2}^{n+1}\lambda_{\mathbb{S}^2}(d\mathbf{x}_i)\nu(dm_i)\right)\nu(da),\\
	\intertext{where $(\mathbf{x}_1,m_1)=(\mathbf{o},a)$. Define $I_n=\int_{\{B_{\mathbb{S}^2}(\mathbf{o},r)\times E\}^n}\xi_{n}\{(\mathbf{x}_1,m_1),\dots,(\mathbf{x}_n,m_n)\}\prod_{i=1}^{n}\lambda_{\mathbb{S}^2}(d\mathbf{x}_i)\nu(dm_i)$ and pull out all terms containing $(\mathbf{x}_1,m_1)=(\mathbf{o},a)$, we obtain}
	&=1+\frac{1}{\nu(C)}\sum_{n=1}^{\infty}\frac{(-\bar{\rho}_E)^n}{n!}\sum_{\Pi\subset\{1,\dots,n\}}J_{|\Pi|}^{CE}(r) \sum_{k=1}^{n-|\Pi|}\sum_{\substack{E_1,\dots,E_k\neq\emptyset, \text{disjoint}\\ \cup_{j=1}^k E_j=\{1,\dots,n\}\setminus\Pi}}\prod_{j=1}^k I_{|E_j|} ,\\
	\intertext{where $|\cdot|$ denotes cardinality and $\mathcal{P}_n$ is the power set of $\{1,\dots,n\}$. This can be expressed as}
	1-D^{CE}_{\text{inhom}}(r)&=1+\\
	&\frac{1}{\nu(C)}\sum_{n=1}^{\infty}\frac{(-\bar{\rho}_E)^n}{n!}J_{|\emptyset|}^{CE}(r)  \sum_{k=1}^{n-|\emptyset|}\sum_{\substack{E_1,\dots,E_k\neq\emptyset, \text{disjoint}\\ \cup_{j=1}^k E_j=\{1,\dots,n\}\setminus\emptyset} }\prod_{j=1}^k I_{|E_j|} +\\
	&\frac{1}{\nu(C)}\sum_{n=1}^{\infty}\frac{(-\bar{\rho}_E)^n}{n!}J_{|\{1,\dots,n\}|}^{CE}(r) \sum_{k=1}^{n-|\{1,\dots,n\}|}\sum_{\substack{E_1,\dots,E_k\neq\emptyset, \text{disjoint}\\ \cup_{j=1}^k E_j=\{1,\dots,n\}\setminus\{1,\dots,n\}} }\prod_{j=1}^k I_{|E_j|}+\\
	&\frac{1}{\nu(C)}\sum_{n=1}^{\infty}\frac{(-\bar{\rho}_E)^n}{n!}\sum_{\substack{\Pi\subset\{1,\dots,n\}\\\Pi\notin\{\emptyset,\{1,\dots,n\}\}}}J_{|\Pi|}^{CE}(r) \sum_{k=1}^{n-|\Pi|}\sum_{\substack{E_1,\dots,E_k\neq\emptyset, \text{disjoint}\\ \cup_{j=1}^k E_j=\{1,\dots,n\}\setminus\Pi} }\prod_{j=1}^k I_{|E_j|}.\\
	\intertext{Since $\xi_1\equiv 1$ then $J_0^{CE}(r)=\nu(C)$ by construction, and under the convention $\sum_{k=1}^0\cdot=1$, we have}
	&	1-D^{CE}_{\text{inhom}}(r)=1+\\
	&\frac{1}{\nu(C)}\sum_{n=1}^{\infty}\frac{(-\bar{\rho}_E)^n}{n!} \sum_{k=1}^{n}\sum_{\substack{E_1,\dots,E_k\neq\emptyset, \text{disjoint}\\ \cup_{j=1}^k E_j=\{1,\dots,n\}} }\prod_{j=1}^k I_{|E_j|}  + \\
	&\frac{1}{\nu(C)}\sum_{n=1}^{\infty}\frac{(-\bar{\rho}_E)^n}{n!}J_{n}^{CE}(r) + \\
	&\frac{1}{\nu(C)}\sum_{n=1}^{\infty}\frac{(-\bar{\rho}_E)^n}{n!}\sum_{\substack{\Pi\subset\{1,\dots,n\}\\\Pi\notin\{\emptyset,\{1,\dots,n\}\}}}J_{|\Pi|}^{CE}(r) \sum_{k=1}^{n-|\Pi|}\sum_{\substack{E_1,\dots,E_k\neq\emptyset, \text{disjoint}\\ \cup_{j=1}^k E_j=\{1,\dots,n\}\setminus\Pi} }\prod_{j=1}^k I_{|E_j|}.\\
	\intertext{Since there are ${n \choose d}$ ways to select $d$ elements of the set $\{1,\dots,n\}$, the second summand of the final term can be simplified as}
		1-D^{CE}_{\text{inhom}}(r)&=1+\\
	&\frac{1}{\nu(C)}\sum_{n=1}^{\infty}\frac{(-\bar{\rho}_E)^n}{n!} \sum_{k=1}^{n}\sum_{\substack{E_1,\dots,E_k\neq\emptyset, \text{disjoint}\\ \cup_{j=1}^k E_j=\{1,\dots,n\}} }\prod_{j=1}^k I_{|E_j|} + \\
	&\frac{1}{\nu(C)}\sum_{n=1}^{\infty}\frac{(-\bar{\rho}_E)^n}{n!}J_{n}^{CE}(r) + \\
	&\frac{1}{\nu(C)}\sum_{n=1}^{\infty}\frac{(-\bar{\rho}_E)^n}{n!}\sum_{d=1}^{n-1}{n \choose d}J_{d}^{CE}(r) \sum_{k=1}^{n-d}\sum_{\substack{E_1,\dots,E_k\neq\emptyset, \text{disjoint}\\ \cup_{j=1}^k E_j=\{1,\dots,n-d\}} }\prod_{j=1}^k I_{|E_j|}.\\
	\intertext{Switching the order of summation and some basic manipulation gives}
		1-D^{CE}_{\text{inhom}}(r)&=1+\\
	&\frac{1}{\nu(C)}\sum_{n=1}^{\infty}\frac{(-\bar{\rho}_E)^n}{n!} \sum_{k=1}^{n}\sum_{\substack{E_1,\dots,E_k\neq\emptyset, \text{disjoint}\\ \cup_{j=1}^k E_j=\{1,\dots,n\}} }\prod_{j=1}^k I_{|E_j|}  + \\
	&\frac{1}{\nu(C)}\sum_{n=1}^{\infty}\frac{(-\bar{\rho}_E)^n}{n!}J_{n}^{CE}(r) + \\
	&\frac{1}{\nu(C)}\left\{\sum_{n=1}^{\infty}\frac{(-\bar{\rho}_E)^n}{n!}J_{n}^{CE}(r)\right\} \left\{\sum_{l=1}^{\infty} \frac{(-\bar{\rho}_E)^l}{l!}\sum_{k=1}^{l}\sum_{\substack{E_1,\dots,E_k\neq\emptyset, \text{disjoint}\\ \cup_{j=1}^k E_j=\{1,\dots,l\}} }\prod_{j=1}^k I_{|E_j|} \right\},\\
	\intertext{which can be factorised as,}
	&=\frac{1}{\nu(C)}\left\{\nu(C) + \sum_{n=1}^{\infty}\frac{(-\bar{\rho}_E)^n}{n!}J_{n}^{CE}(r)\right\}\left\{1+\sum_{l=1}^{\infty} \frac{(-\bar{\rho}_E)^l}{l!}\sum_{k=1}^{l}\sum_{\substack{E_1,\dots,E_k\neq\emptyset, \text{disjoint}\\ \cup_{j=1}^k E_j=\{1,\dots,l\}} }\prod_{j=1}^k I_{|E_j|}\right\}\\
	&=J^{CE}_{\text{inhom}}(r)G(1-u^r_{\mathbf{y},E}).
	\end{align*}
This gives the identity. Also by IRWMI $J^{CE}_{\text{inhom}}(\cdot)$ is independent of the typical point and so are $F^E_{\text{inhom}}(\cdot)$ and $D^{CE}_{\text{inhom}}(\cdot)$.
\end{proof}

\section{Proof of Proposition \ref{prop:J:inhom:equal:1}}\label{proof:J:inhom:equal:1}
This proof follows analgously to Proposition 2 of \cite{Cronie2016} but in the context of spherical marked point processes. Prior to proving the proposition, it can be shown by the Campbell formula that when $X_C$ is independent of $X_E$ then
\begin{align}
\rho^{n_C+n_E}&\{(\mathbf{x}_1,m_1),\dots,(\mathbf{x}_{n_C},m_{n_C}),(\tilde{\mathbf{x}}_1,\tilde{m}_1),\dots,(\tilde{\mathbf{x}}_{n_E},\tilde{m}_{n_E})\}=\nonumber\\
&\rho^{n_C}\{(\mathbf{x}_1,m_1),\dots,(\mathbf{x}_{n_C},m_{n_C})\}\rho^{n_E}\{(\tilde{\mathbf{x}}_1,\tilde{m}_1),\dots,(\tilde{\mathbf{x}}_{n_E},\tilde{m}_{n_E})\},\label{eq:density:independent}
\end{align}
where $(\mathbf{x}_i,m_i)\in\mathbb{S}^2\times C$ and $(\tilde{\mathbf{x}}_i,\tilde{m}_i)\in\mathbb{S}^2\times E$. This identity follows due to uniqueness of the $n^{\text{th}}$-order intensity functions.

We now provide the proof of Proposition \ref{prop:J:inhom:equal:1}.
\begin{proof}
	\begin{align*}
	G^{!}_{\mathbf{y},C}(1-u^r_{\mathbf{y},E})&=1+\frac{1}{\nu(C)}\sum_{n=1}^{\infty}\frac{(-\bar{\rho}_E)^n}{n!}\\
	&\int_C\int_{\{B_{\mathbb{S}^2}(\mathbf{o},r)\times E\}^n}\frac{\rho^{(n+1)}\{(\mathbf{o},a),(\mathbf{x}_1,m_1),\dots,(\mathbf{x}_n,m_n)\}}{\rho(\mathbf{o},a)\rho(\mathbf{x}_1,m_1)\cdots\rho(\mathbf{x}_n,m_n)}\prod_{i=1}^n\lambda_{\mathbb{S}^2}(d\mathbf{x}_i)\nu(dm_i)\nu(da).\\
	\intertext{By (\ref{eq:density:independent}) we have}
	G^{!}_{\mathbf{y},C}(1-u^r_{\mathbf{y},E})&=1+\frac{1}{\nu(C)}\sum_{n=1}^{\infty}\frac{(-\bar{\rho}_E)^n}{n!}\\
	&\int_C\int_{\{B_{\mathbb{S}^2}(\mathbf{o},r)\times E\}^n}\frac{\rho(\mathbf{o},a)\rho^{(n)}\{(\mathbf{x}_1,m_1),\dots,(\mathbf{x}_n,m_n)\}}{\rho(\mathbf{o},a)\rho(\mathbf{x}_1,m_1)\cdots\rho(\mathbf{x}_n,m_n)}\prod_{i=1}^n\lambda_{\mathbb{S}^2}(d\mathbf{x}_i)\nu(dm_i)\nu(da)\\
	&=1+\sum_{n=1}^{\infty}\frac{(-\bar{\rho}_E)^n}{n!}\\
	&\int_{\{B_{\mathbb{S}^2}(\mathbf{o},r)\times E\}^n}\frac{\rho^{(n)}\{(\mathbf{x}_1,m_1),\dots,(\mathbf{x}_n,m_n)\}}{\rho(\mathbf{x}_1,m_1)\cdots\rho(\mathbf{x}_n,m_n)}\prod_{i=1}^n\lambda_{\mathbb{S}^2}(d\mathbf{x}_i)\nu(dm_i)\\
	&=G(1-u^r_{\mathbf{y},E}).
	\end{align*}
	Hence the equality is proven and $J^{CE}_{\text{inhom}}(r)\equiv 1$ when $X_C$ and $X_E$ are independent.
\end{proof}

\section{Proof of Proposition \ref{prop:unbiased:inhom}}\label{proof:unbiased:inhom}

\begin{proof}
	This proof follows analagously to the proof of Lemma 1 by \cite{Cronie2016} but in the context of spherical point processes instead of Euclidean ones. The proof of unbiasedness for $K^{CE}_{\text{inhom}}(\cdot)$ holds trivially by taking expectations of the estimator and equating Definition \ref{def:K:inhom}.

	For $\hat{F}^{E}_{\text{inhom}}(\cdot)$, by taking expectations we have,
	\begin{align*}
	\mathbb{E}\left\{1-\hat{F}^{E}_{\text{inhom}}(r)\right\} &= \frac{1}{|I_{W_{\ominus r}}|}\sum_{\mathbf{p}\in I_{W_{\ominus r}}} \mathbb{E}\left( \prod_{(\mathbf{x},m)\in X} \left[1-\frac{\bar{\rho}_E \mathbbm{1}\{d_{\mathbb{S}^2}(\mathbf{p},\mathbf{x})\leq r, m\in E\}}{\rho(\mathbf{x},m_\mathbf{x})}\right]\right)\\
	&=\frac{1}{|I_{W_{\ominus r}}|}\sum_{\mathbf{p}\in I_{W_{\ominus r}}} G(1-u^r_{\mathbf{p},E}).\\
	\intertext{By Theorem \ref{thm:J:inhom:generating}, the generating functional is independent of the point $\mathbf{p}$ and so,}
	&=\frac{1}{|I_{W_{\ominus r}}|}\sum_{\mathbf{p}\in I_{W_{\ominus r}}} G(1-u^r_{\mathbf{o},E})\\
	&=G(1-u^r_{\mathbf{o},E}),
	\end{align*}
	hence the unbiasedness of $\hat{F}^{E}_{\text{inhom}}(\cdot)$.

	For $\hat{D}^{CE}_{\text{inhom}}(\cdot)$, by taking expectations we have,
	\begin{align*}
	\lambda_{\mathbb{S}^2}&(W_{\ominus r})\nu(C)\mathbb{E}\left\{1-\hat{D}^{CE}_{\rm inhom}(r)\right\}\\
	&=\mathbb{E}\left(\sum_{(\mathbf{x},m_{\mathbf{x}})\in X} \frac{\mathbbm{1}\{(\mathbf{x},m_{\mathbf{x}})\in W_{\ominus r}\times C\}}{\rho(\mathbf{x},m_\mathbf{x})}\prod_{(\mathbf{y},m_{\mathbf{y}})\in X}\left[1- \frac{\bar{\rho}_E\mathbbm{1}\{d_{\mathbb{S}^2}(\mathbf{x},\mathbf{y}) < r,m_{\mathbf{y}}\in E\}}{\rho(\mathbf{y},m_{\mathbf{y}})}\right]\right).\\
	\intertext{By the Campbell-Mecke Theorem,}
	&=\int_{W_{\ominus r}}\int_C \mathbb{E}^!_{\mathbf{x},m_{\mathbf{x}}}\left(\frac{1}{\rho(\mathbf{x},m_\mathbf{x})}\prod_{(\mathbf{y},m_{\mathbf{y}})\in X}\left[1- \frac{\bar{\rho}_E\mathbbm{1}\{d_{\mathbb{S}^2}(\mathbf{x},\mathbf{y}) < r,m_{\mathbf{y}}\in E\}}{\rho(\mathbf{y},m_{\mathbf{y}})}\right]\right)\rho(\mathbf{x},m_{\mathbf{x}})\lambda_{\mathbb{S}^2}(d\mathbf{x})\nu(dm_{\mathbf{x}})\\
	&=\int_{W_{\ominus r}}\int_C \mathbb{E}^!_{\mathbf{x},m_{\mathbf{x}}}\left(\prod_{(\mathbf{y},m_{\mathbf{y}})\in X}\left[1- \frac{\bar{\rho}_E\mathbbm{1}\{d_{\mathbb{S}^2}(\mathbf{x},\mathbf{y}) < r,m_{\mathbf{y}}\in E\}}{\rho(\mathbf{y},m_{\mathbf{y}})}\right]\right)\lambda_{\mathbb{S}^2}(d\mathbf{x})\nu(dm_{\mathbf{x}})\\
	&=\nu(C)\int_{W_{\ominus r}}\left[\frac{1}{\nu(C)}\int_C \mathbb{E}^!_{\mathbf{x},m_{\mathbf{x}}}\left(\prod_{(\mathbf{y},m_{\mathbf{y}})\in X}\left[1- \frac{\bar{\rho}_E\mathbbm{1}\{d_{\mathbb{S}^2}(\mathbf{x},\mathbf{y}) < r,m_{\mathbf{y}}\in E\}}{\rho(\mathbf{y},m_{\mathbf{y}})}\right]\right)\right]\lambda_{\mathbb{S}^2}(d\mathbf{x})\\
	&=\nu(C)\int_{W_{\ominus r}}G^!_{\mathbf{x},C}(1-u^r_{\mathbf{x},E})\lambda_{\mathbb{S}^2}(d\mathbf{x}).\\
	\intertext{Theorem \ref{thm:J:inhom:generating} dictates the generating functional is independent of the typical point $\mathbf{x}$, and so}
	&=\lambda_{\mathbb{S}^2}(W_{\ominus r})\nu(C)G^!_{\mathbf{o},C}(1-u^r_{\mathbf{o},E}).
	\end{align*}
	This gives the unbiasedness of $D^{CE}_{\text{inhom}}(\cdot)$. Ratio unbiasedness of $J^{CE}_{\text{inhom}}(\cdot)$ follows trivially from unbiasedness of $F^{E}_{\text{inhom}}(\cdot)$ and $D^{CE}_{\text{inhom}}(\cdot)$, and the assumption that $\nu$ is known.
\end{proof}

\section{Proof of Proposition \ref{prop:rotation:inhom}}\label{proof:rotation:inhom}

\begin{proof}
	This proof follows analogously to the proof of Proposition 4 by \cite{Cronie2016}, except here in the context of rotations of spherical point processes. We show that the estimators can be written as a function of $(\Xi^{X_C},\Xi^{X_E})$. Then, by the independence of $X_C$ and $X_E$, and rotational invariance of $\Xi$ we have that $(\Xi^{X_C},\Xi^{X_E})\stackrel{d}{=}(\Xi^{X_C},\Xi^{X_E}_O)$ for any $O\in\mathcal{O}(3)$.

	For any $\mathbf{x}\in\mathbb{S}^2$, by the local finiteness of $X$ and an inclusion-exclusion argument, we have that
	\begin{align*}
	\prod_{(\mathbf{y},m_{\mathbf{y}})\in X}&\left[1- \frac{\bar{\rho}_E\mathbbm{1}[(\mathbf{y},m_{\mathbf{y}})\in \{B_{\mathbb{S}^2}(\mathbf{x},r)\times E\}]}{\rho(\mathbf{y},m_{\mathbf{y}})}\right]\\
	&=1+\sum_{n=1}^{\infty}\frac{(-\bar{\rho}_E)^n}{n!}\left[\sum^{\neq}_{(\mathbf{x_1},m_{1}),\dots,(\mathbf{x_n},m_{n})\in X}\prod_{i=1}^n\frac{\mathbbm{1}[(\mathbf{y},m_{\mathbf{y}})\in \{B_{\mathbb{S}^2}(\mathbf{x},r)\times E\}]}{\rho(\mathbf{y},m_{\mathbf{y}})}\right]\\
	&=1+\sum_{n=1}^{\infty}\frac{(-\bar{\rho}_E)^n}{n!}\xi_{X_E}^{[n]}\{B_{\mathbb{S}^2}(\mathbf{x},r)\},
	\end{align*}
	where $\xi_{X_E}^{[n]}\{B_{\mathbb{S}^2}(\mathbf{x},r)\}=\sum^{\neq}_{(\mathbf{x_1},m_{1}),\dots,(\mathbf{x_n},m_{n})\in X}\prod_{i=1}^n\frac{\mathbbm{1}[(\mathbf{y},m_{\mathbf{y}})\in \{B_{\mathbb{S}^2}(\mathbf{x},r)\times E\}]}{\rho(\mathbf{y},m_{\mathbf{y}})}$ is the factorial power measure of $\Xi^{X_C}$, and hence is a function of $\Xi^{X_C}$. This implies that $\hat{F}^{E}_{\text{inhom}}(\cdot)$ can be written as a function of $\Xi^{X_C}$ since,
	\begin{align*}
	1-\hat{F}^{E}_{\text{inhom}}(r) &= \frac{1}{|I_{W_{\ominus r}}|}\sum_{\mathbf{p}\in I_{W_{\ominus r}}} \prod_{(\mathbf{x},m)\in X} \left[1-\frac{\bar{\rho}_E \mathbbm{1}\{d_{\mathbb{S}^2}(\mathbf{p},\mathbf{x})\leq r, m\in E\}}{\rho(\mathbf{x},m_\mathbf{x})}\right]\\
	&= \frac{1}{|I_{W_{\ominus r}}|}\sum_{\mathbf{p}\in I_{W_{\ominus r}}} \left[1+\sum_{n=1}^{\infty}\frac{(-\bar{\rho}_E)^n}{n!}\xi_{X_E}^{[n]}\{B_{\mathbb{S}^2}(\mathbf{p},r)\}\right].
	\end{align*}

	Since we assume $C$ and $E$ are disjoint, and by the local finiteness of $X$, we have $\lambda_{\mathbb{S}^2}(W_{\ominus r})\nu(C)\{1-\hat{D}^{CE}_{\rm inhom}(r)\}$ can be written as,
	\begin{equation*}
	\int_{W_{\ominus r} \times C}\left(1+\sum_{n=1}^{\infty}\frac{-\bar{\rho}_E^n}{n!}\xi_{X_E}^{[n]}\{B_{\mathbb{S}^2}(\mathbf{x},r)\}\right)\Xi^{X_C}(d\mathbf{x},dm)
	\end{equation*}

	Thus in the event that $\nu$ is known, both $\hat{F}^{E}_{\text{inhom}}(\cdot)$ and $\hat{D}^{CE}_{\rm inhom}(\cdot)$ can be written as a function of $(\Xi^{X_C},\Xi^{X_E})$ and hence $\hat{F}^{E}_{\text{inhom}}(\cdot)$, $\hat{D}^{CE}_{\rm inhom}(\cdot)$, and $\hat{J}^{CE}_{\rm inhom}(\cdot)$ are distributionally invariant under rotations of either $\Xi^{X_C}$ or $\Xi^{X_E}$.

	In the event that $\nu$ is unknown, we can use the estimator (\ref{eq:stoyan:stoyan:estimator}). In this case it can be rewritten as,
	\begin{equation*}
	\sum_{(\mathbf{x},m)\in X\cap (W_{\ominus r}\times C)} \frac{1}{\rho(\mathbf{x},m)}=\Xi^{X_C}(W_{\ominus r}\times C).
	\end{equation*}
	Hence, we still retain rotational distributional invariance of $\hat{F}^{E}_{\text{inhom}}(\cdot)$, $\hat{D}^{CE}_{\rm inhom}(\cdot)$, and $\hat{J}^{CE}_{\rm inhom}(\cdot)$.

	Furthermore, note that $\hat{K}^{CE}_{\text{inhom}}(\cdot)$ can be written as,
	\begin{equation*}
	\hat{K}^{CE}_{\text{inhom}}(r) = \frac{1}{\lambda_{\mathbb{S}^2}(W_{\ominus r})\nu(C)\nu(E)}\int_{W_{\ominus r} \times C}\Xi^{X_E}\{B_{\mathbb{S}^2}(\mathbf{x},r)\}\Xi^{X_C}(d\mathbf{x},dm),
	\end{equation*}
and hence is distributionally invariant under rotations of either $\Xi^{X_C}$ or $\Xi^{X_E}$. As before, in the event that $\nu$ is unknown we can use the estimator (\ref{eq:stoyan:stoyan:nu:estimator}) which is a function of $\Xi$. For example,
	\begin{align*}
	\widehat{\nu(C)} &= \frac{1}{\lambda_{\mathbb{S}^2}(W_{\ominus r})}\sum_{(\mathbf{x},m)\in X\cap (W_{\ominus r}\times C)} \frac{1}{\rho(\mathbf{x},m)}\\
	&=\frac{1}{\lambda_{\mathbb{S}^2}(W_{\ominus r})}\Xi^{X_C}(W_{\ominus r}\times C).
	\end{align*}
	Therefore 
	\begin{align*}
\frac{1}{\lambda_{\mathbb{S}^2}(W_{\ominus r})}\Xi^{X_C}(W_{\ominus r}\times C)\Xi^{X_E}(W_{\ominus r}\times E)
	\end{align*}
	is an estimator for $\lambda_{\mathbb{S}^2}(W_{\ominus r})\nu(C)\nu(E)$, and hence still maintains distributional invariance under rotations of $\Xi$.
\end{proof}

\bibliographystyle{chicago}
\bibliography{ref}

\end{document}